\documentclass[pra,twocolumn,a4paper,nofootinbib]{revtex4-1}

\usepackage{hyperref}
\usepackage{graphicx}
\usepackage{amsmath}
\usepackage{amssymb}
\usepackage{color}
\usepackage{amsthm}
\usepackage{amsfonts}
\usepackage{kbordermatrix}
\newtheorem{theorem}{Theorem}
\newtheorem{cor}{Corollary}

\newtheorem{lem}{Lemma}

\begin{document}

\title{Locality-Preserving Logical Operators in Topological Stabiliser Codes}
\author{Paul Webster}
\affiliation{Centre for Engineered Quantum Systems, School of Physics, The University of Sydney, Sydney, NSW 2006, Australia}
\author{Stephen D.~Bartlett}
\affiliation{Centre for Engineered Quantum Systems, School of Physics, The University of Sydney, Sydney, NSW 2006, Australia}

\date{24 January 2018}

\begin{abstract}
Locality-preserving logical operators in topological codes are naturally fault-tolerant, since they preserve the correctability of local errors.  Using a correspondence between such operators and gapped domain walls, we describe a procedure for finding all locality-preserving logical operators admitted by a large and important class of topological stabiliser codes.  In particular, we focus on those equivalent to a stack of a finite number of surface codes of any spatial dimension, where our procedure fully specifies the group of locality-preserving logical operators.  We also present examples of how our procedure applies to codes with different boundary conditions, including colour codes and toric codes, as well as more general codes such as abelian quantum double models and codes with fermionic excitations in more than two dimensions.  
\end{abstract}

\maketitle
\section{Introduction}
Quantum computation offers the potential for algorithms that can solve important problems significantly more quickly than is possible on classical computers~\cite{Shor,Grover}. The realisation of such computation, however, is made challenging by the tendency of large quantum systems to decohere very quickly~\cite{Chuang}.  To combat decoherence, it is expected that quantum computers will need to make use of quantum error correcting codes, allowing for the effects of small amounts of decoherence to be corrected before it causes irreversible damage to the computation~\cite{Shor2}.  In addition, however, we must also demand that logical operators---gates---can be performed fault-tolerantly, meaning that the gates must not amplify correctable errors to uncorrectable ones~\cite{Preskill, NC,Gottesman2}. A particular quantum error correcting code allows only a subset of logical operators to be implemented fault-tolerantly~\cite{EastinKnill}.  This subset will generally be different for different choices of codes.  A central question in the theory of quantum computation concerns how to identify the fault-tolerant logical operators associated with a given quantum error correcting code, and how to identify codes for which the fault-tolerant gates have desirable properties.

Topological stabiliser codes are a class of quantum error correcting codes that have attracted widespread attention due to their simplicity and their relation to robust zero-temperature phases of quantum many-body systems~\cite{KitaevA,Bombin2DTSC,Niggetal}.  The codestates of such codes are topologically protected, meaning that all local errors are correctable~\cite{BDcP}. As such, logical operations will be fault tolerant provided they are locality-preserving~\cite{BK}. Such locality-preserving logical operators (LPLOs) may be seen as a natural generalisation of transversal operators~\cite{Terhal}: errors can grow as long as they remain bounded by some constant size.  However, Bravyi and K{\"o}nig~\cite{BK}, and subsequently Pastawski and Yoshida~\cite{PY}, have shown that topological stabiliser codes have strong restrictions on the class of LPLOs that may be implemented fault-tolerantly.  Specifically, LPLOs in topological stabiliser codes must lie within a bounded set of levels of the Clifford hierarchy.  

Recently, the structure of LPLOs in topological stabiliser codes has been further developed through a correspondence between gapped domain walls and LPLOs~\cite{YoshidaA}. Gapped domain walls are topological structures associated with a symmetry that permutes anyons (and other topological objects)  \cite{KitaevKong,Lanetal}.  This insight is intriguing, as it suggests that the identification of all LPLOs in a topological code may be determined via the structure of the topological data of the model alone.

In this paper, we exploit this correspondence between gapped domain walls and LPLOs in topological stabiliser codes to construct a framework for identifying and classifying the set of LPLOs for a given topological stabiliser code.  Specifically, we detail a procedure to find all of the LPLOs for a topological stabiliser code that is locally equivalent to a finite number of toric codes. This is a large class of codes, containing all non-chiral, translationally invariant two dimensional topological stabiliser codes~\cite{BDcP}, and a wide range of higher dimensional topological stabiliser codes including all colour codes~\cite{Kubica}. More generally, we discuss how the results can also be extended to allow for analysis of abelian quantum double models~\cite{KitaevA,YoshidaC}, and higher dimensional models with fermionic excitaitons~\cite{LW}.

The paper is structured as follows. In Sec.~\ref{II}, we introduce topological stabiliser codes, and describe the structure of the particular class of such codes we focus on: those that are locally equivalent to a finite number of toric codes. In Sec.~\ref{III}, we review the connection between LPLOs and gapped domain walls for two dimensional topological stabiliser codes, following Refs.~\cite{YoshidaA,Beverland}. We provide a more detailed classification of the LPLOs admitted by such codes, and explicitly illustrate how the different boundary conditions of toric, surface and colour codes affect the admitted logical gates.  In Sec.~\ref{IV}, we begin exploring the more exotic behaviour of higher dimensional codes. In particular, we discuss how the concepts of domain walls and excitations must be generalised to account for the higher dimensional excitations encountered in codes of more than two dimensions. We illustrate these concepts by the three dimensional examples of a stack of disjoint surface codes and the colour code. We also explore how the ideas may also be applied to the three dimensional Levin-Wen fermion model~\cite{LW}, which is not equivalent to a number of toric codes but which has a similar structure that nonetheless admits the same type of analysis. Finally, in Sec.~\ref{V}, we generalise these ideas to fully classify the LPLOs admitted by any code locally equivalent to a finite number of disjoint surface codes. We also generalise our examples of colour codes and toric codes to see how these results may be adapted to codes of different boundary conditions.

\section{Topological Stabiliser Codes}
\label{II}
For concreteness and simplicity, we will restrict our consideration to topological stabiliser codes. Such codes have logical states that are topologically protected ground states of a Hamiltonian, $\mathcal{H}=\sum_{i} S_i$, where $\{S_i\}$ is a set of local, commuting Pauli operators. The abelian group, $\mathcal{S}$, generated by $\{S_i\}$ is thus the stabiliser group of the logical states of the code. We note, however, that our results are expected to apply more broadly to a wide range of topological models, including quantum double models and string-net models, where the Hamiltonian consists of a sum of local commuting projectors.  (We will briefly return to this generalisation later.)

In addition, we will focus on topological stabiliser codes that are locally equivalent to some finite number of toric codes of the same spatial dimension.  While this may seem to be a strong restriction, it is, in fact, known to include a wide range of codes of interest.  Specifically, it includes all non-chiral, translationally invariant two dimensional topological stabiliser codes \cite{BDcP}, as well as a wide range of relevant higher dimensional codes, including all colour codes \cite{Kubica}.

The bulk properties of a toric code can be specified by two parameters, up to local equivalence: the spatial dimension $d$, and the dimension of its magnetic flux excitations, $M$.  Note that $M$ can take any value in the range $0$ to $d-2$, and that it fixes the dimension of electric charge excitations as $E=d-2-M$.  Therefore, on an infinite lattice, there are $d-1$ inequivalent toric codes of dimension $d$.  

To specify the logical qubits (i.e.,~the topological ground state degeneracy) of such codes, we need to specify the boundary conditions.  The canonical set of boundary conditions we will consider for a $d$-dimensional toric code are those of the \emph{surface code} (and its higher dimensional generalizations)~\cite{SC}, as shown in Figs.~\ref{fig:2DSC} and \ref{fig:3DSC}.  These boundary conditions correspond to those of a $d$ dimensional hypercubic lattice with $M+1$ pairs of opposite $d-1$ dimensional \emph{smooth} boundaries, meaning that magnetic fluxes can condense at them. The other $E+1$ pairs of boundaries are then \emph{rough}, meaning that electric charges can condense at them.  This choice of boundary conditions gives a two dimensional ground state degeneracy---a single logical qubit.  A choice of anticommuting pair of logical operators is given by an $\bar{X}$ logical operator defined by an $M+1$ dimensional hyperplane of $X$'s mapped out by a magnetic flux condensing on a pair of opposite smooth boundaries, and a $\bar{Z}$ operator similarly realised but with electric charges and rough boundaries.  We refer to this specification of $d$-dimensional toric code with boundary conditions as a \emph{$d$-dimensional surface code}, as it naturally generalises the two dimensional surface code.

\begin{figure}
\centering
\includegraphics[width=0.8\linewidth]{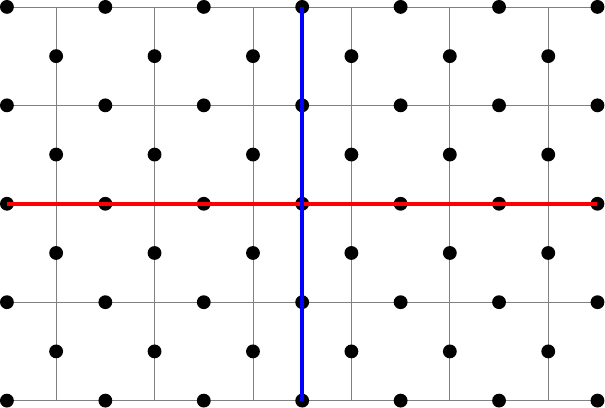}
\caption{Two dimensional surface code: The top and bottom edges are smooth, the left and right are rough. The blue line shows an $\bar{X}$ operator that may be viewed equivalently as being produced by a magnetic flux following the path, or by $X$ operators being applied at each qubit along the path. The red line shows a $\bar{Z}$ operator, which may similarly be considered the path of an electric charge, or a product of $Z$ operators. \label{fig:2DSC}}
\label{fig:SC}
\end{figure}

\begin{figure}
\centering
\includegraphics[width=0.5\linewidth]{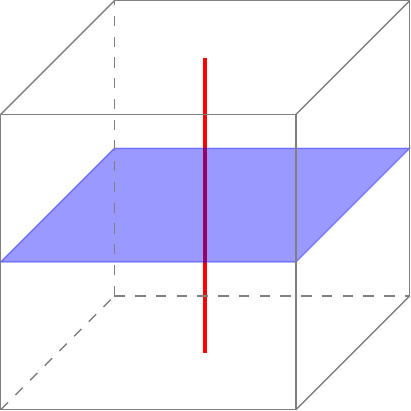}
\caption{Three dimensional surface code: The top and bottom faces are rough, the other four smooth. The red line shows a logical $Z$ operator produced by a zero dimensional electric charge following the path. The blue plane is a logical $X$ operator produced by a one dimensional magnetic flux condensing at one of the four smooth faces, and travelling to the opposite face, tracing out the plane behind it. \label{fig:3DSC}}
\end{figure}

The $d$-dimensional surface code can serve as a building block for more general codes.  One such generalization of a $d$-dimensional topological stabiliser code is a \emph{stack} of $n$ disjoint $d$-dimensional surface codes.  There is no requirement for the codes in the stack to be identical; they may have magnetic fluxes of different dimensions. Such a code is then parameterised by its dimension, $d$, and the numbers $\{n_i\}$ of surface codes it has with $i$ dimensional magnetic fluxes, for $0 \leq i \leq d-2$. Since the codes are disjoint, the ordering of the codes in the stack is not important. We refer to such a code of this type as a \emph{stacked code}, specified by the set of parameters $\{d,n_i\}$. It clearly encodes $n=\sum_{i=0}^{d-2} n_i$ logical qubits; one for each surface code.

The codes we consider here, then, are locally equivalent to a stacked code in the bulk. Specifically, a $d$-dimensional code which is locally equivalent to $n$ toric codes, of which $n_i$ have $i$ dimensional magnetic fluxes must be locally equivalent to the stacked code with the same parameters in the bulk. This implies that the bulk properties of the code must be the same as those for the corresponding stacked code. We can add additional structure by considering boundary conditions that are more general than those for the stacked code. For example, a colour code on a triangular lattice may be viewed as two toric codes which are connected by a fold along one of the boundaries~\cite{Kubica}. Thus, properties of the code which depend on the boundary conditions of the code may differ from those of the corresponding stacked code.

The general approach we take throughout the paper is to consider the stacked code to be a canonical code for each set of parameters  $\{d,n_i\}$. We then classify the relevant bulk properties of codes by performing this classification on stacked codes. Specifically, the properties we deduce from the stacked code are the excitations of the code, their exchange and braiding statistics, and the domain walls the code admits. This has the advantage that these properties are especially easy to determine for stacked codes. Properties of the code which are sensitive to boundary conditions, in particular the logical operators, can then be deduced by considering the excitations and domain walls of the code in the context of the boundary conditions.

\section{Classification for Two Dimensional Codes}
\label{III}

In this section, we establish the key concepts and correspondences necessary to classify the LPLOs in the context of two dimensional topological stabiliser codes, where they are well-understood. In Sec.~\ref{subsec:2DLPLO}, we review the correspondence between LPLOs and gapped, transparent domain walls in two dimensional topological codes. This correspondence has previously been described by Yoshida~\cite{YoshidaA}, but we reproduce it here in our framework since it is central to the results we develop in the remainder of the paper. Where possible, we develop these ideas in a way that is independent of dimension, and so they will be relevant to our consideration of higher dimensional codes in subsequent sections.  In Sec.~\ref{subsec:2Dexamples}, we illustrate how the ideas we have developed may be used to classify the LPLOs admitted by a given two dimensional code. 

\subsection{Logical Operators and Gapped Domain Walls}
\label{subsec:2DLPLO}

\subsubsection{Locality-Preserving Logical Operators}

\begin{figure}
\centering
\includegraphics[width=0.95\linewidth]{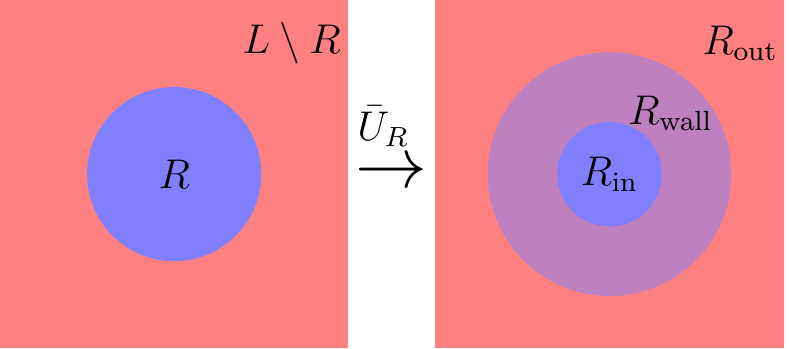}
\caption{A domain wall corresponding to LPLO $\bar{U}$. On the left the lattice, $L$, is partitioned into region $R$ and the rest. After $\bar{U}_R$, the restriction of $\bar{U}$ to $R$, has acted the image of the part of the Hamiltonian confined to $R$ extends outwards to a region, $R'$ larger than $R$ by only a finite width, $C$, and the rest of the Hamiltonian similarly has image with support on a region $L \setminus R''$ where $R''$ is smaller than $R$ by only a finite width, $C$. Thus, the purple region on the right, ${R}_{\text{wall}} \equiv R'\setminus R''$ may be interpreted as a domain wall, since by acting on the original system by the restriction of $\bar{U}$ to $R$ we create an inner region ${R}_{\text{in}}\equiv R''$ transformed as though $\bar{U}$ had acted on the whole of $L$ and an outer region which is unaffected by the action, separated by the region ${R}_{\text{wall}}$.}
\label{fig:corresp}
\end{figure}

In a topological model, a \emph{logical operator} is a unitary operator that maps the topologically degenerate ground space onto itself. 
A \emph{locality-preserving logical operator} (LPLO) is a logical operator that preserves the locality of errors. Specifically, a logical operator $\bar{U}$ is locality-preserving if and only if there exists some constant $C$ such that, for any operator of a code with support only in some region of the code, $R$, we have that $\bar{U}R\bar{U}^\dag$ has support only in a region $R'$ which is at most $C$ larger than $R$. LPLOs in topological stabiliser codes always include Pauli logical operators, which act on errors only by adding phases. Since they are always locality preserving, however, we omit them from explicit consideration in our subsequent analysis. Instead our focus is on LPLOs that transform errors by more than just a phase. Such operators must be two dimensional, since they must transform errors that could occur at any point in the lattice.

Such LPLOs provide a natural way to introduce \emph{domain walls}.  Consider $\bar{U}_R$, the restriction of an LPLO to a simply connected region $R$ of the code. The first observation we can make is that this gives rise only to a region of width at most $2C$ in which the ground space of the code is not preserved. To see this, partition the Hamiltonian of the code into the sum of two parts, $\mathcal{H} = \mathcal{H}_{\text{in}} + \mathcal{H}_{\text{out}}$.  The interior Hamiltonian $\mathcal{H}_{\text{in}}$ includes only terms with support entirely within $R$, and $\mathcal{H}_\text{out} \equiv \mathcal{H}-\mathcal{H}_{\text{in}}$.  The action of our LPLO $\bar{U}$ on $\mathcal{H}$ is given by
\begin{equation}
\bar{U} \mathcal{H}\bar{U}^\dag = \bar{U} \mathcal{H}_{\text{in}}\bar{U}^\dag + \bar{U} \mathcal{H}_{\text{out}}\bar{U}^\dag \,.
\end{equation}
Now, $\bar{U} \mathcal{H}_{\text{in}}\bar{U}^\dag$ must be confined to a region $R'$ at most $C$ larger than $R$.  Similarly, $\bar{U} \mathcal{H}_\text{out} \bar{U}^\dagger$ has no support on a region $R'' \subseteq R$ that is at most $C$ smaller than $R$.  As a result, if we apply the restriction $\bar{U}_R$ of this LPLO only to $R$, then the interior of $R''$ will be indistinguishable (using any local operator on $R''$) from the case where the unrestricted operator $\bar{U}$ had been applied to the entire code. Therefore, it will be in the ground space of the code, since $\bar{U}\mathcal{H}_{\text{out}}\bar{U}^\dag$ has no support in $R''$. Similarly, the exterior of $R'$ will be unaffected by $\bar{U}_R$ since $\bar{U}\mathcal{H}_{\text{in}}\bar{U}^\dag$ has no support outside $R'$. Thus, as claimed, we have only a region $R' \setminus R''$ of width $2C$ in which the ground space of the code is disturbed. This is illustrated in Fig.~\ref{fig:corresp}. We now argue that this region can be understood as a domain wall.

\subsubsection{Domain Walls}

Domain walls separate two regions of a code that differ by some symmetry of the ground space. In particular, the region $R' \setminus R''$ as defined above can be viewed as a domain wall, since it separates an interior region which has been transformed by $\bar{U}$ from an exterior, untransformed region.	A domain wall created in this way will also be gapped; that is, the Hamiltonian will remain gapped at the wall, since the operator creating it is unitary

Gapped domain walls can be characterised by the behaviour of excitations as they cross the domain wall. {To see this, consider an excitation $a_i$ corresponding to an error that anticommutes with a subset, $\mathcal{S}_i$, of the stabiliser group, $\mathcal{S}$. Upon crossing a domain wall corresponding to $\bar{U}$, it will be transformed to an excitation $a_j$ corresponding to an error that anticommutes with $\mathcal{S}_j = \bar{U}\left(\mathcal{S}_i\right) \subset \mathcal{S}$.} As a concrete example, consider the case where $\bar{U}=\bar{H}$, a logical Hadamard gate on a single surface code. An electric charge $e$, which is a violation of an $X$-type stabiliser in the Hamiltonian, will be transformed to an excitation that is a violation of a $Z$-type stabiliser, that is, a magnetic flux.  Conversely, a magnetic flux crossing this domain wall will be transformed to an electric charge. Note that a domain wall of this type must be \emph{transparent}, meaning that excitations cannot condense from the vacuum at them, since $\bar{U}$ must fix the identity under conjugation. 

The correspondence between LPLOs and gapped, transparent domain walls is bijective for stacked codes. To see this, note first that there exists a unique domain wall associated with each LPLO. Specifically, we have seen that this domain wall will be that constructed by applying a restriction of the LPLO to a simply connected region of the code.

Conversely, a gapped transparent domain wall can be used to implement an LPLO, as follows. Consider a small such wall condensed out of the vacuum. We may then apply local operators which apply the symmetry by which the interior and exterior of the wall differ, to gradually grow the domain wall. Once it has grown to cover the full code then we will have applied a logical operator to the code. Specifically, a wall that transforms excitation $a_i$ to $a_j$ must have stabilisers transformed by $S_i$ to $S_j$ in its interior. Moreover, this operator must be locality-preserving. To see this, consider an error confined to a region, $R$, of the code. Now, grow the domain wall to contain all of the code except for this region, and only then grow the domain wall over this region. Then, the error cannot grow until the domain wall begins to close over it. At this point, all of the Hamiltonian apart from this region and the domain wall of width $2C$ has already been transformed as though the corresponding logical operator had been applied everywhere. Thus, the error can only grow to a region of width $2C$ greater than $R$, and so the logical operator is locality preserving. The LPLO corresponding to a domain wall is unique, since the unique way to implement a logical operator from the wall is to push it onto the boundaries of each surface code in which it has support.

\subsubsection{Finding Domain Walls and Logical Operators}
\label{BrEx}
\begin{figure}
\centering
\includegraphics[width=0.95\linewidth]{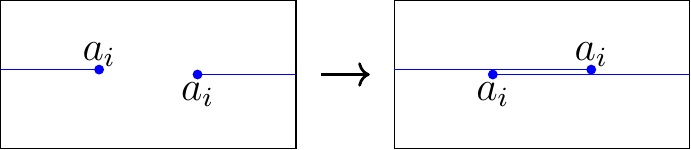}
\caption{Exchange of two excitations $a_i$ that exist as endpoints of a string of $U$ operators, shown in blue. The exchange of such excitations introduces a region between the excitations where $U^2$ is applied, since the strings associated with each excitation overlap. \label{fig:Exc1}}
\end{figure}

\begin{figure}
\centering
\includegraphics[width=0.8\linewidth]{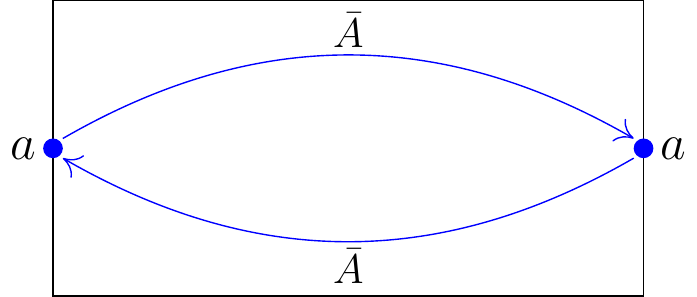}
\caption{Exchange of a pair of excitations $a$ on opposite boundaries of the two dimensional suface code is equivalent to applying the square of the corresponding logical operator, $\bar{A}$. \label{fig:Exc2}}
\end{figure}

The significance of this correspondence is that classifying gapped, transparent domain walls is a relatively easy task, whereas identifying LPLOs may not be. Specifically, gapped, transparent domain walls correspond to permutations of the set of excitations that preserve the structure of the topological model, i.e., of the fusion, exchange and braiding statistics~\cite{KitaevKong,Lanetal}. Ensuring that permutations of excitations preserve the fusion relations can be done automatically by considering the image of an independent generating set of the group of excitations (with fusion as the group multiplication). Exchange and braiding statistics for the codes we consider are encoded in the $T$ and $S$ matrices of the code.

The $T$ matrix encodes the self statistics of the anyonic theory.  It is a diagonal matrix with entries encoding the phase associated with exchanging a pair of identical particles. For example, the element will be 1 if the particle is a boson, or $-1$ for a fermion. In the context of a code, this phase may also be understood in terms of operators that may be applied to the qubits of the code to propagate the excitation. Specifically, assume excitation $a$ exists at an endpoint of a string of unitary operators $U$. Then, as shown in Fig.~\ref{fig:Exc1}, the exchange of two exciations $a$ will cause a string of operators $U^2$ to be applied between the two particles. Consider now the two excitations $a$ originating at opposite boundaries of a code, as shown in Fig.~\ref{fig:Exc2}. The exchange of these particles will then introduce a string of $U^2$ operators across the lattice. Since a string of $U$ operators implements the logical operator, $\bar{A}$, corresponding to excitation $a$, this exchange then applies the operator $\bar{A}^2$. In two dimensional codes, $\bar{A}$ will be a Pauli operator, and so $\bar{A}^2$ will be a phase.

The $S$ matrix, meanwhile, encodes the mutual exchange statistics, or braiding relations. Specifically, the matrix element $S_{i,j}$ encodes the phase associated with braiding particles $a_i$ and $a_j$ around each other, as shown in Fig.~\ref{fig:Braid1}. We may again understand this phase in terms of logical operators. The relevant object is the \textit{group commutator}, $K$ \cite{YoshidaA}. For unitary operators, $U$ and $V$, this is defined as $K(U,V)=UVU^\dag V^\dag$. Indeed, as Fig.~\ref{fig:Braid2} shows, the phase associated with braiding particles $a_i$ and $a_j$ is the group commutator of the corresponding logical operators, $\bar{A}_i$ and $\bar{A}_j$. Since Pauli operators either commute or anticommute, in two dimensional codes this commutator will be $\pm I$.

The interpretations we have provided of elements of the $S$ and $T$ matrices in terms of logical operators provide further insight into the relationship between domain walls and LPLOs. Specifically, the requirements that the domain wall preserve exchange and braiding relations can be seen to be equivalent to the requirement that conjugation by the corresponding LPLO, $\bar{U}$ commutes with taking the square of any logical operator, $\bar{A}$, and with taking the group commutator of any pair of logical operators, $\bar{A}$ and $\bar{B}$. Such requirements on logical operators are manifestly necessary, since, for any unitary operators $U$, $A$ and $B$, the following relations hold.
\begin{align}
(UAU^\dag)^2 &=UA^2 U^\dag\\
K(UAU^\dag , UBU^\dag ) &= UK(A,B)U^\dag
\end{align}
While the braiding and exchange relations of the excitations of two dimensional topological stabiliser codes are already well understood, this relationship between excitation statistics and logical operators will prove invaluable when considering the more complicated excitation structures of higher dimensional codes. For this reason, we introduce it here so that we may illustrate its working in the simpler, two dimensional case.

\begin{figure}
\centering
\includegraphics[width=0.5\linewidth]{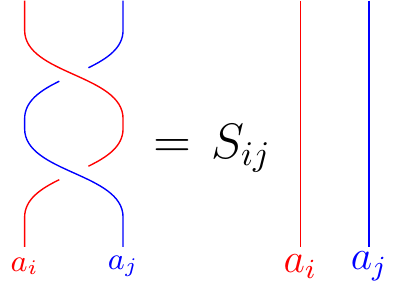}
\caption{Space (horizontal) time (vertical) diagram showing the world lines of two distinct anyons, $a_i$ and $a_j$, being braided. This braiding gives rise to a phase, $S_{ij}$. \label{fig:Braid1}}
\end{figure}

\begin{figure}
\centering
\includegraphics[width=0.9\linewidth]{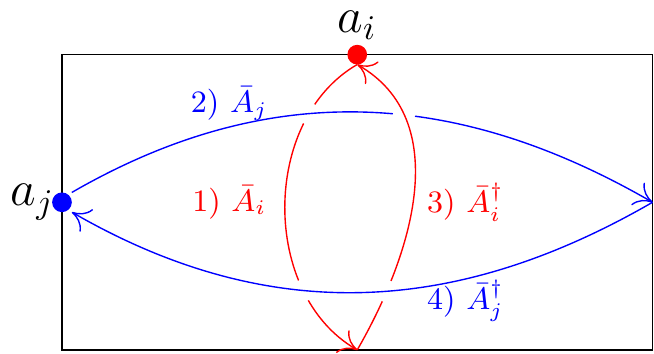}
\caption{Braiding of a pair of anyons $a_i$, $a_j$ in the two dimensional suface code is equivalent to applying the group commutator of their corresponding logical operators, $K(\bar{A}_i,\bar{A}_j) =\bar{A}_i\bar{A}_j\bar{A}_i^\dag \bar{A}_j^\dag$. \label{fig:Braid2}}
\end{figure}

\subsection{Examples in two dimensions}
\label{subsec:2Dexamples}
Here, we provide several examples in two dimensions of this equivalence between gapped domain walls and LPLOs.  We begin by considering the surface code and stacked codes, and then demonstrate how the domain walls we find in this case can also be used to determine the locality preserving gates for codes with other boundary conditions such as the colour code and the toric code.  Again, we follow Yoshida~\cite{YoshidaA}.

\subsubsection{The Surface Code}
As a first example, we consider the simplest stacked code: a single surface code.  We will identify all of the possible gapped domain walls in this code, and then use them to identify LPLOs.  To identify the domain walls, we first construct the $S$ and $T$ matrices, which encode the braiding and exchange relations of the the anyon excitations of the code.  The single surface code's anyonic properties are completely specified by the properties of three excitations:  bosons $e$ and $m$, and a fermion $em$.  Any distinct pair of $e$, $m$ and $em$ excitations have non-trivial braiding \cite{Kitaev}. Moreover, these braiding relations ensure that $em$ particles give a phase of $-1$ under exchange, since such an exchange is equivalent to the braiding of an $e$ particle with an $m$. These relations fix the $S$ and $T$ matrices to be as follows.
\begin{align} \label{eq:ST2DSC}
S &=   \kbordermatrix{  & 1 & e & m & em \\  1 & 1 & 1 & 1 & 1 \\
      e & 1 & 1 & -1 & -1 \\
      m & 1 & -1 & 1 & -1\\
     em & 1 & -1 & -1 & 1} \,\\
     T 
     &=\kbordermatrix{ & 1 & e & m & em \\
     \label{eq:ST2DSC2}
      1 & 1 & 0 & 0 & 0 \\
      e & 0 & 1 & 0 & 0 \\
      m & 0 & 0 & 1 & 0\\
     em & 0 & 0 & 0 & -1}\,
\end{align}

We may also consider an alternative approach to finding these matrices, by considering the corresponding logical operator relations discussed in Sec.~\ref{BrEx}. Specifically, we can identify $(1,e,m,em) \leftrightarrow (\bar{I},\bar{Z},\bar{X},\bar{Z}\bar{X})$. If we enumerate the anyons $a_i$ and denote their corresponding logical operators by $U_i$, we can then express elements of the $S$ and $T$ matrices in terms of logical operators as follows.
\begin{align}
S &=   \kbordermatrix{ & 1 & e & m & em \\
      1 & K(I,I) & K(I,X) & K(I,Z) & K(I,XZ) \\
      e & K(X,I) & K(X,X) & K(X,Z) & K(X,XZ) \\
      m & K(Z,I) & K(Z,X) & K(Z,Z) & K(Z,XZ)\\
     em & K(XZ,I) & K(XZ,X) & K(XZ,Z) & K(XZ,XZ)}\\
     T &=\kbordermatrix{ & 1 & e & m & em \\
      1 & I^2 & 0 & 0 & 0 \\
      e & 0 & X^2 & 0 & 0 \\
      m & 0 & 0 & Z^2 & 0\\
     em & 0 & 0 & 0 & (ZX)^2}
\end{align}
Evaluating these matrices, we indeed find the same results as in Eq.~\ref{eq:ST2DSC} and Eq.~\ref{eq:ST2DSC2}. The generalisation of this alternative approach to higher dimensional codes with more complicated excitation structures will prove invaluable in Sec.~\ref{IV} and \ref{V}.

With the $S$ and $T$ matrices identified, we now identify domain walls by finding automorphisms of anyons that preserve braiding and exchange relations. These correspond to tunnelling matrices that commute with the $S$ and $T$ matrices \cite{Lanetal}. It is simple to see that there are two such walls, the trivial wall $W_0$, and a nontrivial wall $W_1$ with tunnelling matrix:
\begin{equation}
W_1=    \kbordermatrix{ & 1 & e & m & em \\
      1 & 1 & 0 & 0 & 0 \\
      e & 0 & 0 & 1 & 0 \\
      m & 0 & 1 & 0 & 0\\
     em & 0 & 0 & 0 & 1}\,.
\end{equation}
The domain wall $W_1$ interchanges $e$ and $m$ type excitations. Since these excitations correspond to logical operators $\bar{Z}$ and $\bar{X}$ respectively, we may thus conclude that the logical operator, $\bar{U}_1$, corresponding to the wall is that which acts as follows.
\begin{align}
\bar{U}_1 \bar{X}\bar{U}_1^\dag &=\bar{Z} \,,\\
\bar{U}_1 \bar{Z}\bar{U}_1^\dag &=\bar{X} \,.
\end{align}
That is, $\bar{U}_1$ is the logical Hadamard operator. Thus, we may conclude that the only LPLO admissible in a single two dimensional surface code is the logical Hadamard operator. 

\subsubsection{The Stacked Code}
\label{sec:2DStack}
We now consider a more general two dimensional stacked code, consisting of a stack of $n$ surface codes. To begin, we again construct the $S$ and $T$ matrices of the code. To do this, observe that excitations in distinct surface codes will have trivial braiding relations, and so the matrices will be tensor products of the single surface code matrices.
\begin{align}
S &=   \left(\kbordermatrix{  & 1 & e & m & em \\  1 & 1 & 1 & 1 & 1 \\
      e & 1 & 1 & -1 & -1 \\
      m & 1 & -1 & 1 & -1\\
     em & 1 & -1 & -1 & 1}\right)^{\otimes n} \label{eq:S}\\
     T 
     &=\left(\kbordermatrix{ & 1 & e & m & em \\
      1 & 1 & 0 & 0 & 0 \\
      e & 0 & 1 & 0 & 0 \\
      m & 0 & 0 & 1 & 0\\
     em & 0 & 0 & 0 & -1}\right)^{\otimes n} \label{eq:T}
\end{align}
Again, the results of this approach agrees with the matrices we would get by associating anyons $e_j$, $m_j$ from code $j$ in the stack with the Pauli operators $Z_j,X_j$ of this code. This is because, for $U,V,U',V'$ acting on spaces of the same dimension, we may decompose group commutators and squares as
\begin{align}
K(U \otimes V, U' \otimes V') &= K(U,U')\otimes K(V,V') \,,\\
(U \otimes V)^2 &= U^2\otimes V^2\,.
\end{align}

We can then use these $S$ and $T$ matrices to find the full set of gapped, transparent domain walls admitted by the code. For the case $n=2$, we can do this explicitly, and find a group of 72 domain walls \cite{YoshidaA}, generated by the following walls.
\begin{align}
h_1: &\, e_1 \leftrightarrow m_1\\
h_2: &\, e_2 \leftrightarrow m_2\\
s_{12}: &\, m_1 \to m_1e_2, m_2 \to e_1 m_2
\end{align}
Note that here, and throughout, we describe domain walls by their action on electric charges and magnetic fluxes, and omit those charges and fluxes which are unaffected by the wall. These walls correspond to logical Hadamard operators on each of the two surface codes in the stack, and a logical controlled-Z operator between the codes, as can be seen by the following actions on logical Pauli operators:
\begin{align}
\bar{H}_1: &\, \bar{Z}_1 \leftrightarrow \bar{X}_1\\
\bar{H}_2: &\, \bar{Z}_2 \leftrightarrow \bar{X}_2\\
\overline{CZ}_{12}: &\, \bar{X}_1 \to \bar{X}_1\bar{Z}_2, \bar{X}_2 \to \bar{Z}_1 \bar{X}_2
\end{align}

For $n>2$, this generalises to the following group of gapped, transparent domain walls.
\begin{align}
h_i: &\, e_i \leftrightarrow m_i\\
s_{ij}: &\, m_i \to m_ie_j, m_j \to e_i m_j
\end{align}
and corresponding LPLOs
\begin{align}
\bar{H}_i: &\, \bar{Z}_i \leftrightarrow \bar{X}_i\\
\overline{CZ}_{ij}: &\, \bar{X}_i \to \bar{X}_i\bar{Z}_j, \bar{X}_j \to \bar{Z}_i \bar{X}_j
\end{align}
To see this, note that domain walls cannot map $e$ or $m$ type excitations to $em$ excitations, since the exchange statistics of $em$ excitations differ from those of $e$ and $m$. This corresponds to the requirement that real Pauli operators must remain real under conjugation by LPLOs of the code. Thus, the LPLOs of the code must be contained in the real Clifford group, which is generated by Hadamard and controlled-$Z$ operators \cite{Gajewski}. Indeed, since we have seen that the Hadamard and controlled-$Z$ operators are locality preserving, the set of LPLOs is exactly the real Clifford group.

\subsubsection{The Colour Code}

We now consider our first example which is not a stacked code: the two dimensional colour code on a triangular lattice. Such a code is locally equivalent to a surface code which is folded onto itself, or equivalently, to two triangular surface codes that share a nontrivial common boundary~\cite{Kubica}, as shown in Fig.~\ref{fig:2DCC}. Along this common boundary, an anyon from one code can be transformed into an anyon of the same type from the other code. The other boundaries are rough for code one and smooth for code two, and vice versa. The code encodes a single logical qubit.

These boundary conditions establish an equivalence relation between the anyon theories of the two codes. In particular, excitations that may be interchanged at a boundary joining the two codes are considered equivalent. In this case, these boundary conditions mean that $e_1$ and $e_2$ are equivalent excitations (now denoted just $e$), and similarly for the magnetic fluxes.  This equivalence relation has a non-trivial effect on how the domain walls relate to LPLOs.  The domain walls in this code will be the same as those found for the case of two stacked surface codes, but some domain walls will now be equivalent in the sense that they correspond to the same LPLO.

Consider the domain walls $h_1$ and $h_2$, identical to those in the $n=2$ stacked code.  Both of these domain walls interchange $e$ and $m$, and so they both correspond to the same LPLO, $\bar{H}$, acting as a logical Hadamard on the single encoded qubit of the colour code.  The $s_{12}$ domain wall that maps $m$ to $em$ now corresponds to a logical operator mapping $\bar{X}$ to $\bar{Y}\propto\bar{Z}\bar{X}$, since we do not distinguish the codes from which the $e$ and $m$ come in the image of $m$. The LPLOs are thus as follows, where we specify only mappings of $\bar{X}$ and $\bar{Z}$ that are non-trivial.
\begin{align}
\bar{H}: &\, \bar{Z} \leftrightarrow \bar{X}\\
\bar{R}_2: &\, \bar{X} \to \bar{Y}
\end{align}
These operators generate the single qubit Clifford group \cite{Terhal}. We follow Yoshida's convention of referring to the gate that interchanges $\bar{X}$ and $\bar{Y}$ as $\bar{R}_2$ \cite{YoshidaA}, since the notation is easy to generalise to the gates that will arise when we consider higher dimensional colour codes. Specifically, we throughout denote by $R_k$ the single qubit operator which is diagonal with elements $1$ and $e^{2\pi i/k}$ with respect to the computational basis. Note, however, that this $\bar{R}_2$ gate is the phase gate, which is often also denoted by $P$ or $S$ in the literature.

If we consider a stack of colour codes, then it is clear that we will also get $\overline{CZ}$ operators that act between each pair of codes. Indeed, this corresponds to the $\overline{CZ}$ operators that existed for a stacked code, acting between surface codes that are not folded into the same colour code. Thus, a stack of colour codes admits a locality-preserving implementation of the full Clifford group, generated by $\bar{H}$, $\bar{R}_2$, and $\overline{CZ}$, as was known previously via specific transversal implementations of these logical gates~\cite{Kubica}.

\begin{figure}
\centering
\includegraphics[width=0.95\linewidth]{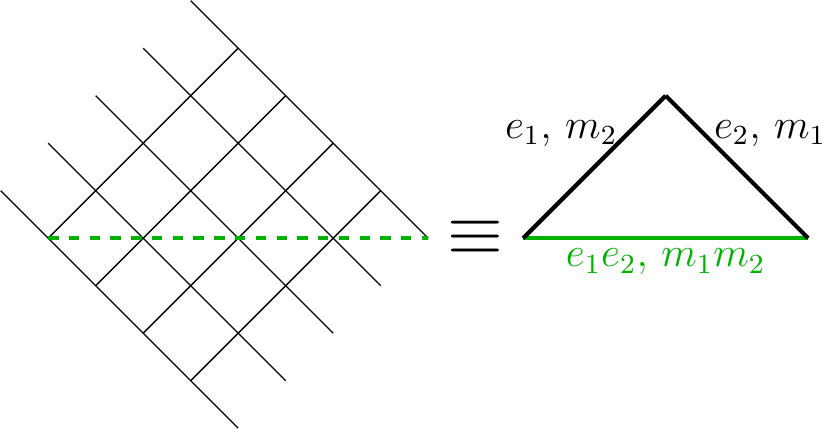}
\caption{Boundary conditions to encode a single qubit in the colour code. The green dashed line on the left is where the surface code should be folded to attain the colour code on the right. Such a fold gives two boundaries where there is a rough edge folded over a smooth one, labelled by the excitations that can condense there, where there is a rough edge folded over or under a smooth one. The third boundary is the common boundary where two anyons of the same type may condense together, or, equivalently, where excitations from different codes may be interchanged. \label{fig:2DCC}}
\end{figure}

\subsubsection{The Toric Code}
As a final two dimensional example, we may consider a genuine toric code, i.e., with boundary conditions corresponding to a torus. Such a code encodes two logical qubits, corresponding to anyons following topologically distinct and non-trivial loops around the torus. We then have logical operators $\bar{Z}_1$ and $\bar{X}_2$ corresponding to excitations following one of the classes of non-trivial loops, and $\bar{Z}_2$ and $\bar{X}_1$ corresponding to them following the other, as shown in Fig.~\ref{fig:TC}.

Growing a domain wall across a code will preserve paths followed by excitations, but can interchange the excitations. Thus, the $h$ domain wall that we have in the surface code corresponds to a logical operator which interchanges $\bar{Z}_1 \leftrightarrow \bar{X}_2$ and $\bar{X}_1 \leftrightarrow \bar{Z}_2$, and so corresponds to a Hadamard operator on both logical qubits followed by a SWAP operator on the qubits. This is therefore the only LPLO admitted by the code.

\begin{figure}
\centering
\includegraphics[width=0.8\linewidth]{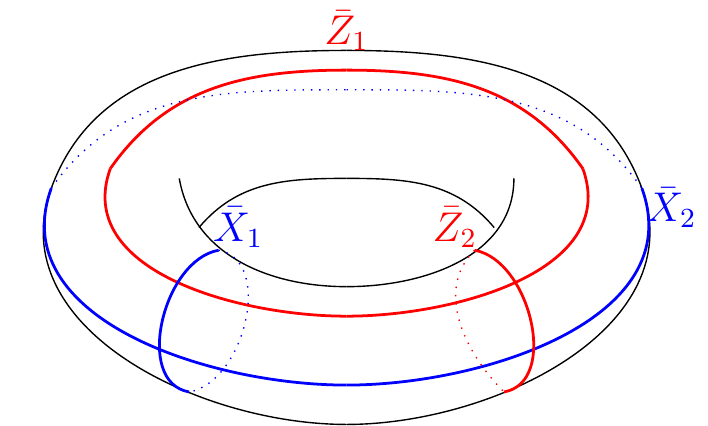}
\caption{$\bar{X}$ (blue) and $\bar{Z}$ (red) operators of the toric code. Since there are two pairs of such operators, the code encodes two logical qubits. \label{fig:TC}}
\end{figure}

\section{Extending to Higher Dimensions: Concepts and Examples}
\label{IV}
We now turn to extending these connections between domain walls and LPLOs to higher dimensional topological stabiliser codes.  As first illustrated by Yoshida for the 3D colour code~\cite{YoshidaA}, higher dimensional codes have a much richer structure, both in terms of their excitations and the generalization of domain walls.  In this section, we will outline several aspects of this additional structure, and use a range of three dimensional codes as examples to demonstrate the relationships with LPLOs.  Our complete analysis of the relationship between domain walls and LPLOs in topological stabiliser codes of higher dimension is left to Sec.~\ref{V}.

\subsection{New Concepts}
\subsubsection{Higher Dimensional Excitations}
The first source of additional structure in higher-dimensional topological stabiliser codes is the dimensionality of low-energy excitations.  Excitations in two dimensional topological codes are point-like anyons.  For higher dimensional codes, excitations can be extended objects of one or more dimensions, such as loops and surfaces.  Much like anyons, these excitations can possess non-trivial topological braiding relationships exhibited through sophisticated effects such as three loop braiding \cite{WangLevin}.

A natural starting point to encapsulating these relations between excitations is by generalising the exchange and braiding relations we considered in two dimensions. Specifically, exchange statistics of excitations, and braiding of pairs of excitations, whether point-like or higher dimensional, may be expected to still be preserved by domain walls in higher dimensions. This expectation is justified by the  fact that the corresponding logical operator actions; taking the square of an operator, and taking the group commutator of a pair of operators, must still commute with conjugation by a logical operator in higher dimensions. Indeed, to avoid the complications of considering the differing geometries of braiding and exchange involving higher dimensional excitations, we will simply consider the squares and commutators of logical operators directly, without seeking to explicitly construct $S$ and $T$ matrices that fully describe the statistics of all the excitations of the code. We will, however, implicitly consider elements of such matrices, by denoting the commutator associated with braiding of excitations $a$ and $b$ by $S_{a,b}$, and the square associated with exchange of a pair of excitations $a$ by $T_{a,a}$.

We may expect that more complex processes, such as three loop braiding, must also be considered in higher dimensional codes \cite{WangLevin,YoshidaA}. Such processes correspond to nested commutators \cite{YoshidaA,YoshidaC}. In fact, however, the preservation of braiding and exchange relations of eigenstate excitations, corresponding to the square and commutator of Pauli operators, implies that such a nested commutator is trivial for the codes we consider. This is because the commutator of two Pauli operators must be a phase, and so the commutator of the image of these Pauli operators must also be this phase. Thus, this commutator commutes with all other operators, giving a trivial nested commutator. This same argument can also be extended to show that all relevant higher level nested commutators will also be trivial. Note that three loop braiding relations involving non-eigenstate excitations can be non-trivial, however. Indeed, Yoshida has showed that such relations can provide insight into what combinations of excitations may condense from the vacuum at opposite sides of a domain wall \cite{YoshidaA}. Using these relations is not essential, however, as the same insight can be attained by comparing the exchange relations of the excitations emerging from each side of the wall.

\subsubsection{Generalised Domain Walls}
\label{IVA2}
To achieve our central goal of classifying LPLOs in higher dimensional codes, we now wish to generalise the correspondence between domain walls and LPLOs that we have in the two dimensional case. To do this, consider first the case of $d$ dimensional LPLOs in $d$ dimensional codes. Such operators will have boundaries as domain walls, just as in the two dimensional case. Moreover, the arguments that gave us a one-to-one correspondence of LPLOs and gapped, transparent domain walls in the two dimensional case will carry through to higher dimensions as well. These domain walls transform excitations in a way that preserves the topological relationships, as in the case of domain walls in two dimensional codes. The only additional requirement that emerges for these domain walls is that they must preserve the dimensionality of excitations that cross them, as shown in Fig.~\ref{fig:2DDW}. This is because the logical operator that gives rise to such a wall must be locality preserving. Thus, a lower dimensional excitation must not grow to a higher dimensional one upon entering the region acted on by that operator. Since the wall must be invertible, this also means that higher dimensional excitations cannot be transformed to lower dimensional ones.

Higher dimensional codes, however, may also admit LPLOs of lower dimension $k<d$. Such operators will have boundaries as $k-1$ dimensional objects with some properties in common with domain walls. The key difference between these objects and true (codimension one) domain walls is that they do not partition the lattice. Instead, they partition a subspace of the lattice. We will refer to these types of objects as \textit{generalised domain walls}, since they are not true domain walls, but share some important properties with them. Since they only partition a subspace of the lattice, point-like excitations may travel from one side of such a generalised domain wall to the other by passing out of this subspace. This is not necessarily true of higher dimensional excitations, however. For example, an extended one dimensional excitation must cross through a one dimensional generalised domain wall at a point, as shown in Fig.~\ref{fig:1DDW}. This generalised domain wall may thus act like a domain wall on the one dimensional excitation that passes through it, but one that only acts at the point on the excitation that actually passes through the generalised domain wall. Thus, such a generalised domain wall must act trivially on a point-like excitation, since such an excitation may avoid it, but can transform an extended one dimensional excitation at a point. Equivalently, we may view this as appending some point-like excitation to the one dimensional excitation.

More generally, a $j$ dimensional excitation passing through a $k$ dimensional boundary of an LPLO in a $d$ dimensional code will intersect with the boundary if and only if $j+k-d-1 \geq 0$. Otherwise the boundary must act trivially on the excitation. If they do intersect, they will do so at a region of dimension $l=j+k-d-1$, and so this is the dimension of the excitation that may be transformed by the generalised domain wall. Thus, $k$-dimensional generalised domain walls correspond to permutations of excitations that preserve the topological relations, and also act on a $j+k-d-1$ dimensional region of each $j$ dimensional excitation.  We conclude that we can classify LPLOs in higher dimensional codes in an analogous way to in two dimensions: by finding automorphisms of code excitations, with the additional condition that allowed automorphisms must act on excitations in a way consistent with a generalised domain wall of a particular dimension.

\begin{figure}
\centering
\includegraphics[width=0.8\linewidth]{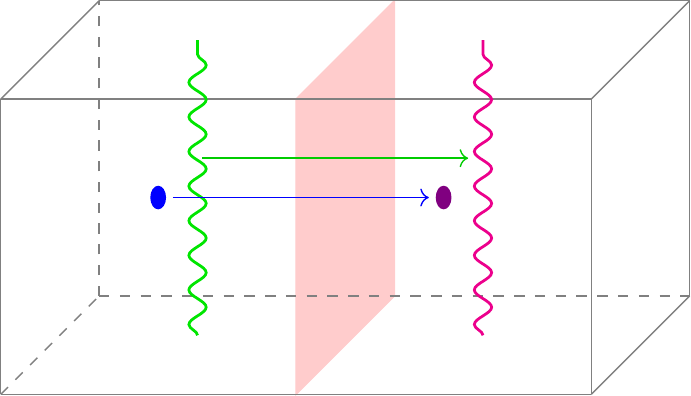}
\caption{A (true) domain wall in a three dimensional code partitions the code into two parts. Zero and one dimensional excitations crossing the wall can be transformed to other excitations of the same dimension. \label{fig:2DDW}}
\end{figure}

\begin{figure}
\centering
\includegraphics[width=0.8\linewidth]{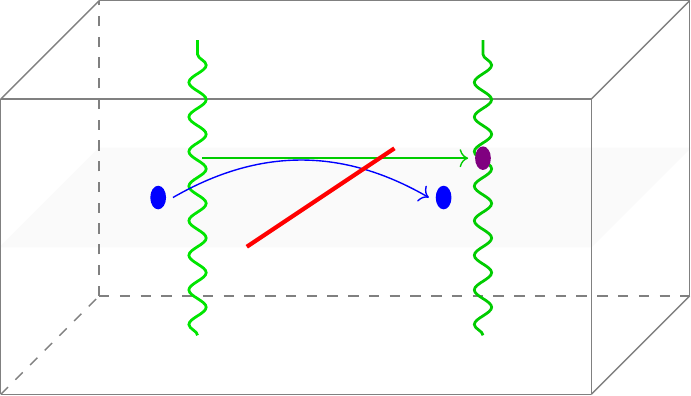}
\caption{A one dimensional generalised domain wall partitions only a two dimensional subspace of a three dimensional code into two parts. Zero dimensional excitations may thus travel from side of the wall to the other without crossing it, while one dimensional excitations extended orthogonally to the wall across the whole lattice must cross it at a point. \label{fig:1DDW}}
\end{figure}

\subsubsection{Non-Eigenstate Excitations}
One further complication remains. Specifically, boundaries of LPLOs---domain walls and their lower dimensional counterparts---can themselves be viewed a type of excitation of the code. To see this, consider applying the restriction of an LPLO, $\bar{U}$, to a region of the code. This action will give rise to an excitation, localised to the boundary of this region.  Note that the ``excitation'' will not in general be an energy eigenstate of the code, but rather a superposition.  

Nonetheless, ``excitations'' of this form may be stable despite not having a well-defined energy, as it may not be possible to remove them locally. To distinguish them from the excitations that are energy eigenstates of the code, we refer to them as \textit{non-eigenstate excitations}. In general, removing such a non-eigenstate excitation requires the application of the inverse of the LPLO on the entire interior region.  To understand how this makes them stable, consider applying a two dimensional LPLO to a region of a two dimensional code that is not simply connected, for example a ring around a torus as shown in Fig.~\ref{fig:TCLO}. Such an operator now gives rise to two topologically non-trivial loop non-eigenstate excitations as its boundary. These loops cannot be removed locally, but only by acting in the region between them to bring them together. Thus, these non-eigenstate excitations can only be created or annihilated in pairs, as is the case for electric charges and magnetic fluxes. This implies that a single non-trivial loop non-eigenstate excitation cannot be removed locally, and so will be stable. On a surface code, which exists on a simply connected lattice, we may consider instead the case of such an excitation extending between two opposite boundaries, as shown in Fig.~\ref{fig:SCLO}. Such a non-eigenstate excitation can only be removed by pushing it onto a boundary where it can be absorbed, just as for the eigenstate excitations of the code. 

\begin{figure}
\centering
\includegraphics[width=0.8\linewidth]{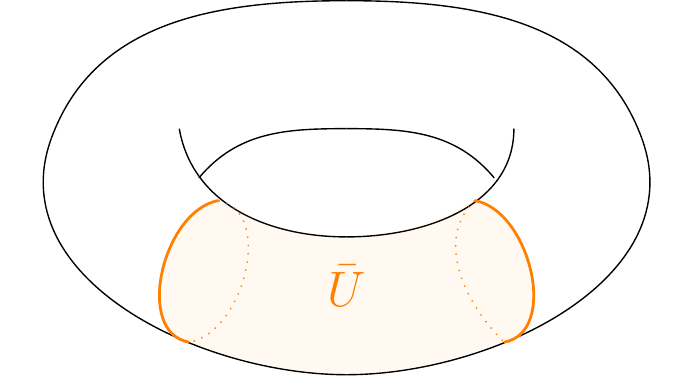}
\caption{Applying the restriction of an LPLO, $\bar{U}$, to a topologically non-trivial region of a torus gives rise to two non-eigenstate excitations. These can only be created or annihilated in pairs, and so a single such non-eigenstate excitation is stable. \label{fig:TCLO}}
\end{figure}

\begin{figure}
\centering
\includegraphics[width=0.7\linewidth]{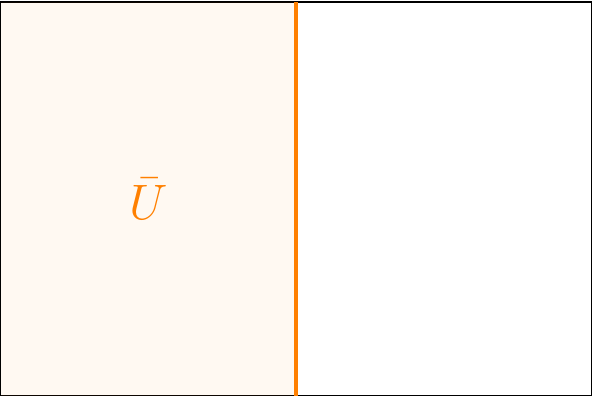}
\caption{Applying the restriction of an LPLO to a region stretching to three boundaries of a surface code gives rise to a single eigenstate excitation that extends between two opposite boundaries, and can only be removed by being pushed on to one of the other boundaries. \label{fig:SCLO}}
\end{figure}

These non-eigenstate excitations are important as they must be considered along with eigenstate excitations of the code as objects that may be permuted by generalised domain walls. For example, we may imagine a generalised domain wall that takes an extended, one dimensional eigenstate excitation to a one dimensional non-eigenstate excitation. We will see examples of such walls when we consider the three dimensional codes below. Note that walls of this type are not possible in two dimensions since all eigenstate excitations in such codes are point-like and so cannot be transformed into extended non-eigenstate excitations by any domain wall. This is significant, since generalised domain walls that only permute eigenstate excitations correspond to Clifford logical operators. Indeed, we refer to them as \textit{Clifford domain walls}. Thus, the fact that two dimensional topological stabiliser codes admit only Clifford domain walls explains why they also admit only Clifford LPLOs. We revisit this relationship between the dimensionality of codes and the Clifford hierarchy in more detail in Sec.~\ref{VA}. We first however, in Sec.~\ref{IVB}, explore through examples the greater range of LPLOs that are possible in three dimensional topological stabiliser codes.

\subsection{Three Dimensional Codes}
\label{IVB}
In this section, we illustrate the new concepts and structures we have described above by exploring several three dimensional codes. We begin by considering stacked codes. We then consider the three dimensional colour code, which may be viewed as three surface codes of three dimensions folded into one another. This code was also considered by Yoshida in \cite{YoshidaA}, but we provide a more complete treatment of it, and show how it fits into our broader framework. We also consider an example of a code which is not locally equivalent to any number of toric codes, but which has a sufficiently similar structure to still allow for analysis; the Levin-Wen fermion model.

\subsubsection{Three Dimensional Surface Code}
We consider first a single three dimensional surface code. Recall that such a code has one type of eigenstate excitation as point-like, while the other is an extended one dimensional excitation. For simplicity, we will assume that it is the magnetic flux that is one dimensional, while the electric charge is point-like. We still identify electric charges crossing the lattice with $\bar{Z}$ operators and magnetic fluxes with $\bar{X}$ operators. Thus, all the braiding and exchange relations of eigenstate excitations coincide with the same group commutators and squares as they did in two dimensions. This means that the same constraints that we had on domain walls in the two dimensional surface code must also be satisfied by Clifford domain walls in the three dimensional code. Recall that the only non-trivial domain wall admitted by the two dimensional code was $h: e \leftrightarrow m$. Thus, this is the only candidate for a non-trivial Clifford domain wall in the three dimensional code. In this case, however, the electric charges and magnetic fluxes are of different dimensions. Since no domain wall may increase the dimension of an excitation that crosses it, therefore, the $h$ wall cannot be realised in the three dimensional code. Thus, we conclude that there are no non-trivial domain walls admitted by the three dimensional surface code. This further implies that no LPLOs are admitted by this code, or indeed by any other code that is locally equivalent to the surface code in the bulk, such as the three dimensional toric code.

\subsubsection{Three Dimensional Stacked Code}
\label{IVB2}
A far more interesting example comes if we consider instead a larger stack of $n$ three dimensional surface codes. Considering first only eigenstate excitations, we have that the same braiding and exchange relations will hold as do for the two dimensional stacked code. This is because the eigenstate excitations correspond to the same Pauli logical operators as in two dimensions. The same braiding and exchange constraints thus apply to domain walls in this code as do for the two dimensional stacked code, and thus candidate domain walls in this code must be products of the following walls we calculated in section \ref{sec:2DStack}.
\begin{align}
h_i: &\, e_i \leftrightarrow m_i\\
s_{ij}: &\, m_i \to m_ie_j, m_j \to e_i m_j
\end{align}

To see which of these candidate walls are indeed allowed, we must now consider, in turn, which permutations may be realised as two dimensional domain walls, which as one dimensional (generalised) domain walls and which are inconsistent with walls of either dimension. Consider first two dimensional walls. These may interchange different point-like excitations and also interchange different one dimensional excitations. Specifically, it may interchange different electric charges and interchange different magnetic fluxes. The subset of the candidate walls that are of this form are generated by walls $c_{ij}$: $m_i \to m_im_j$, $e_j \to e_i e_j$. To see this, note that such a wall corresponds to the logical operator $\overline{CNOT}_{ij}$: $\bar{X}_i \to \bar{X}_i \bar{X}_j$, $\bar{Z}_j \to \bar{Z}_i \bar{Z}_j$ which, along with Pauli operators, generates the group of operators that preserve both of the sets of Pauli operators consisting of only $I$ and $X$ type operators, and only $I$ and $Z$ type operators, respectively \cite{Delfosse}. Next, consider one dimensional generalised domain walls. Such walls may append point-like electric charges to one dimensional magnetic fluxes, but must act trivially on electric charges. The group of such domain walls are generated by $s_{ij}$ walls. As for the two dimensional code, these walls correspond to control-$Z$ operators, $\overline{CZ}_{ij}$. The action of this corresponding operator is shown in Fig.~\ref{3DCZ}. 

Thus, the Clifford domain walls admitted by a three dimensional stacked code are generated by the following walls, where we relabel $s_{ij}$ to $s_{ij}^{(2)}$ in anticipation of results in the following paragraphs.
\begin{align}
c_{ij}:&\, m_i \to m_im_j,\, e_j \to e_i e_j\\
s_{ij}^{(2)}:&\, m_i \to m_ie_j,\, m_j \to e_i m_j
\end{align}

We may now consider these generalised domain walls as (non-eigenstate) excitations. Specifically, the $c_{ij}$ walls are two dimensional excitations. Since this is a higher dimension than any of the eigenstate excitations, they cannot be mapped to by eigenstate excitations crossing any domain wall. Thus, this type of wall does not yield any further domain walls or LPLOs. The $s_{ij}^{(2)}$ type wall, however, is one dimensional. Since this is of the same dimension as the magnetic fluxes of the code, there is a possibility of having two dimensional domain walls that map magnetic fluxes to non-eigenstate excitations. Note that we refer to such domain walls as ``non-Clifford'' domain walls, since we will see that they correspond to LPLOs that are outside the Clifford group.

To determine if this possibility is realised, we consider the types of non-Clifford domain walls that are not forbidden by dimensional considerations, and test if they preserve relevant braiding and exchange relations. Note that since we have walls $w_{ij}=c_{ij}c_{ji}c_{ij}: e_i \leftrightarrow e_j,\, m_i \leftrightarrow m_j$ that swap the codes from which excitations come, we may assume our walls act trivially on electric charges. This is because any wall which interchanges different electric charge excitations will then differ from the walls we consider by a product of $w_{ij}$ walls.  Now, consider first a wall that maps $m_i \to s^{(2)}_{j,k}$, where $j \neq k$, but we make no assumptions about $i$. Compare the braiding phases $S_{m_i,e_i}=K(Z_i,X_i)=-1$ and $S_{s^{(2)}_{j,k},e_i}=K(CZ_{j,k},Z_i)=1$. For the domain wall to preserve braiding statistics we would require these two braiding phases to be equal. Thus, that they are not equal implies domain walls of type $m_i \to s^{(2)}_{j,k}$ cannot be allowed.

The only other type of possible wall maps a magnetic flux to a composite excitation of a magnetic flux and $s$ type excitation. Specifically, consider such a wall that again acts trivially on electric charges, but maps $m_{i} \to m_l s^{(2)}_{jk}$. Now consider again the braiding phase $S_{m_i, e_i}=-1$, and compare it to $S_{m_l s^{(2)}_{j,k}, e_i}= K(X_l CZ_{j,k}, Z_i)=(-1)^{\delta_{i,l}}$. We thus conclude that the preservation of braiding relations requires that $l=i$, and so we need consider only domain walls that act as $m_i \to m_i s^{(2)}_{j,k}$. Next consider $T_{m_i,m_i}= (X_i)^2=1$ compared to $T_{m_is_{j,k},m_is_{j,k}}= (X_iCZ_{j,k})^2=1$ if and only if $i \neq j,k$. This comparison implies that $i \neq j,k$ and so any possible non-Clifford domain walls must be of the form $m_i \to m_i s^{(2)}_{j,k}$ up to multiplication by Clifford walls. Since we have seen that this preserves exchange relations, it remains only to verify that it also preserves braiding relations to conclude that it is indeed an allowed wall. Since it preserves electric charges it will preserve braiding relations between them. Thus, we need only consider that $S_{m_is^{(2)}_{jk},e_l}=K(X_i CZ_{j,k}, Z_l)= CZ_{j,k} K(X_i, Z_l)CZ_{j,k}=K(X_i,Z_l)=S_{m_i,e_l}$ and 
\begin{align}
S_{m_{i}s^{(2)}_{jk},m_{j}s^{(2)}_{ki}} &= K(X_j CZ_{ki}, X_i CZ_{jk}) \notag\\
&= CZ_{ki}X_j CZ_{jk}X_i X_j CZ_{ki} X_i CZ_{jk} \notag\\
&= CZ_{ki}X_j X_i CZ_{jk}X_j CZ_{jk} CZ_{ki}X_i \notag\\
&= X_j CZ_{ki} X_i X_j Z_k CZ_{ki} X_i \notag\\
&= X_j CZ_{ki}X_i CZ_{ki} X_j Z_k X_i \notag \\
&= X_j X_iZ_k X_jZ_kX_i \notag \\
&=X_i^2X_j^2Z_k^2\notag \\
&=1\notag \\
&= S_{m_{i},m_j}
\end{align}
So, braiding relations are indeed preserved and so the wall is allowed.

Thus, we may conclude that the generalised domain walls of the code are generated by the following three types of walls.
\begin{align}
c_{ij}:&\, m_i \to m_im_j,\, e_j \to e_i e_j\\
s_{ij}^{(2)}:&\, m_i \to m_ie_j,\, m_j \to e_i m_j\\
s_{ijk}^{(3)}:&\, m_i \to m_i s^{(2)}_{jk},\, m_j \to m_j s^{(2)}_{ki},\, m_k \to m_k s^{(2)}_{ij}
\end{align}
We have already identified that the first two of these walls correspond in the stacked code to $\overline{CNOT}$ and $\overline{CZ}$ operators. The third maps $X_i \to X_i CZ_{jk}$, $X_j \to X_j CZ_{ki}$, $X_k \to X_k CZ_{ij}$, and so is a control-control-$Z$ operator, $\overline{CCZ}$ \cite{YoshidaC}. The LPLOs of the three dimensional stacked code are therefore generated by $\overline{CNOT}, \overline{CZ}, \overline{CCZ}$.

\begin{figure}
\centering
\includegraphics[width=0.95\linewidth]{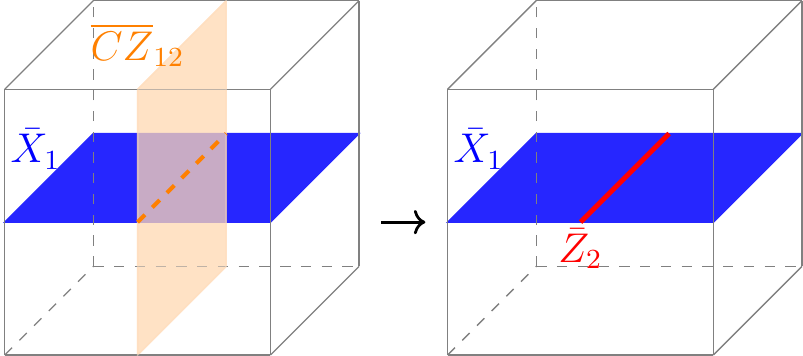}
\caption{The two dimensional $\overline{CZ}_{12}$ in the three dimensional stacked code with three surface codes, which is equivalent to a three dimensional colour code on a cubic lattice \cite{Kubica} acts on the two dimensional Pauli operator, $\bar{X}_1$, by appending a one dimensional $\bar{Z}_2$ operator to it.\label{3DCZ}}
\end{figure}

\subsubsection{Three Dimensional Colour Code}
We now consider three surface codes on tetrahedral lattices in three dimensions attached along a two dimensional boundary. Specifically, we identify the three adjacent faces from each of the codes all with a single face which allows excitations of the form $e_ie_je_k$ and $m_im_j$ to condense, for $i \neq j \neq k$. Excitations of the form $e_ie_j$ may condense along one dimensional edges of the tetrahedral lattice. The code encodes a single logical qubit. Kubica \textit{et al} have shown that it is locally equivalent to the colour code on a tetrahedral lattice \cite{Kubica}.

Similarly to the case of the two dimensional colour code, we may determine the LPLOs admitted by the code by considering excitations that may be interchanged at boundaries of the code to be equivalent. Specifically, we consider all electric charges of the code to be equivalent and label them $e$, and all magnetic fluxes of the code also equivalent, labelled $m$. We can then consider the action of the domain walls we have found on these equivalence classes of excitations. Specifically, the $c_{ij}$ type domain wall only interchanges excitations that are equivalent, and so is equivalent to the trivial wall. Thus, we get no non-trivial logical operator from this wall. The $s_{ij}^{(2)}$ wall, however, appends an electric charge to a magnetic flux. In terms of equivalence classes, therefore, the wall acts as $s^{(2)}: m \leftrightarrow em$. Since magnetic fluxes of the code correspond to the $\bar{X}$ operator on the logical qubit and $em$ to $\bar{Y} \propto \bar{X}\bar{Z}$, the corresponding logical operator for this wall is thus $\bar{R}_2: \bar{X} \leftrightarrow \bar{Y}$. The action of this operator is shown in Fig.~\ref{fig:3DR2}.

We may now consider the $s^{(3)}_{ijk}$ type domain wall. Again, considering excitations of the same type, but from different codes to be equivalent, this wall can be viewed as $s^{(3)}: m \to ms^{(2)}$. This then corresponds to the LPLO $\bar{R}_3: \bar{X} \to \bar{X}\bar{R}_2$ \cite{YoshidaA}. Since $R_3^2=R_2$, the LPLOs are therefore generated by $\bar{R}_3$ alone.

\begin{figure}
\centering
\includegraphics[width=0.95\linewidth]{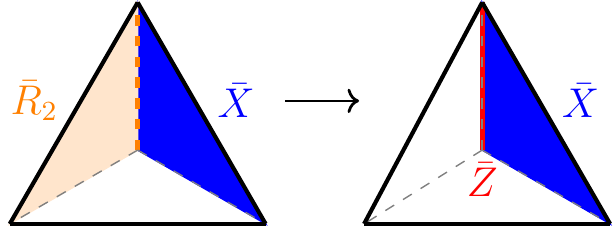}
\caption{The $\bar{R}_2$ LPLO acting on a $\bar{X}$ operator in the three dimensional colour code on a tetrahedral lattice (equivalent to three three-dimensional surface codes attached along a face). \label{fig:3DR2}}
\end{figure}

\subsubsection{Levin-Wen Fermion Model}
\label{Sec:LW}
We now consider an example of a code that is not locally equivalent to any number of toric codes, but nonetheless can be analysed using the framework we have developed. The Levin-Wen fermion model \cite{LW} is a topological stabiliser code which can be defined on a cubic lattice of two qubit sites, of total lattice side length, $L$. The stabilisers are four site operators \cite{Haah}, as shown in Fig.~\ref{fig:LWStab}.

\begin{figure}
\centering
\includegraphics[width=0.95\linewidth]{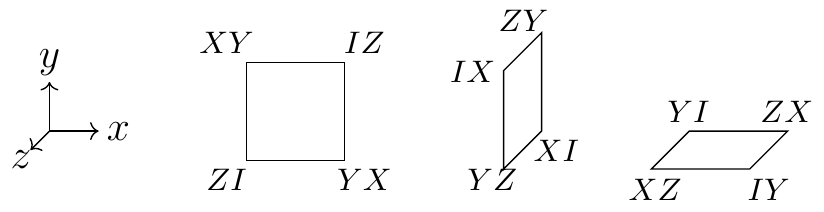}
\caption{The stabilisers of the Levin-Wen Fermion Model are operators that act across the four two qubit sites that lie around each face of a cube in the lattice. There are three different types of stabilisers, each of which act at each face in the lattice of a particular orientation (parallel to the $xy$, $yz$ and $zx$ planes respectively). \label{fig:LWStab}}
\end{figure}

The excitation structure of the code is similar in some ways to that of the three dimensional surface code.  There exists a point-like excitation, which we refer to as an electric charge, and an extended one-dimensional excitation, which we refer to as a magnetic flux. This code differs from the surface code in its exchange statistics, however. Specifically, the electric charge is a fermion, rather than a boson~\cite{LW}. 

The code also has a different logical operator structure depending on the parity of $L$. Specifically, for odd $L$ the code encodes only two logical qubits, while for even $L$ it encodes three logical qubits \cite{Haah}. For this reason, the LPLOs must be considered separately for the different parities of $L$. Since the bulk properties of the code do not depend on the total size of the lattice, however, the two cases will have the same set of generalised domain walls.  They can be found equally well by considering either case.    

Using our framework, we show that the code admits $\overline{CZ}$ operators between logical qubits associated with the same code, and the only type of logical operator acting between different codes in a stack of such codes is a logical $\text{SWAP}$ operator between corresponding qubits of the codes. 

To show this result, we choose to find the domain walls by considering the case of odd $L$. We then have two anticommuting pairs of logical Pauli operators, which may be realised as follows, where by $\hat{x},\hat{y},\hat{z}$ we denote unit vectors directed in the $x,y,z$ directions~\cite{Haah}. Each of the logical operators is only specified up to a phase.
\begin{itemize}
\item $\bar{X}_1$ is realised by $IX$ operators acting at every site along a plane spanned by $\hat{x}$ and $(\hat{y}+\hat{z})$ vectors. $\bar{Z}_1$ is realised by $ZZ$ operators acting at every site along a line oriented parallel to $\hat{z}$.
\item $\bar{X}_2$ is realised by $YX$ operators acting at every site along a plane spanned by $\hat{z}$ and $(\hat{x}+\hat{y})$ vectors. $\bar{Z}_2$ is realised by $XX$ operators acting at every site along a line oriented parallel to $\hat{x}$.
\end{itemize}
As for the surface code, the $\bar{X}$ operators are realised by magnetic fluxes crossing the lattice to trace out the support of the relevant operator, and similarly $\bar{Z}$ operators are realised by electric charges operators crossing the lattice. This, in fact, implies that the $\bar{Z}$ operators as described must be $\pm i\bar{Z}$ operators, so that they square to $-1$ and give the correct exchange statistics of the electric charge. Nonetheless, for simplicity we will continue to include this phase only when explicitly considering exchange statistics, and omit it elsewhere.

By considering dimensions of excitations, we can deduce that the only possible Clifford domain wall for a single Levin-Wen code is $p: m \to em$. Note that this wall is distinct from the $s$ type wall found in a stacked surface code, since $p$ appends an electric charge to a magnetic flux of the same code. To determine if this wall is indeed allowed, we must consider if it preserves the exchange and braiding relations.

To assist with determining exchange relations throughout our analysis of the code, it is worth noting the following identity,
\begin{equation}
 (AB)^2=A^2 K(A^\dag, B) B^2 \label{ab2}.
 \end{equation}
 We can simplify this result in the special case that $A=\omega P$, where $\omega$ a phase, and $P$ is Hermitian (e.g.~Pauli). In that case we have that $K(A^\dag,B)=K(A,B)$, and so Eq.~\ref{ab2} reduces to the following identity,
\begin{equation}
(AB)^2=A^2K(A,B) B^2.
\end{equation}
In particular, the result implies that for an eigenstate excitation $a$, and general excitation $b$, we have the following identity,
\begin{equation}
T_{ab,ab}=(AB)^2=A^2K(A,B) B^2= T_{a,a} S_{a,b} T_{b,b} \label{eq:exchid}.
\end{equation}
In the special case where $a$ and $b$ are both point-like, this identity can be visualised as in Fig.~\ref{fig:exchid}. 

Using the identity in Eq.~\ref{eq:exchid}, we can verify that domain wall, $p$, does indeed preserve exchange statistics, by the following calculation,
\begin{align}
T_{em,em} &=T_{e,e}  S_{e,m} T_{m,m}\\
&=-1  \cdot -1 \cdot T_{m,m}\\
&=T_{m,m}.
\end{align}
Moreover, braiding relations are preserved, as shown by the following calculations,
\begin{align}
S_{em,e}=K(Y,Z)=-1=K(X,Z)=S_{m,e}\\
S_{em,m}=K(Y,X)=K(X,Y)=S_{m,em}.
\end{align}
Thus, the generalised domain wall $p$ is allowed.

\begin{figure}
\centering
\includegraphics[width=0.5\linewidth]{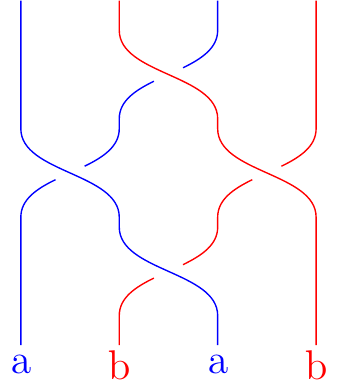}
\caption{Space (horizontal) time (vertical) diagram of the exchange of two composite excitations, $ab$. Focussing on the crossings of world lines of $a$ and $b$ excitations, we can see that it is equivalent to an exchange of $a$ excitations, a braid of $a$ with $b$ and an exchange of $b$ excitations. \label{fig:exchid}}
\end{figure}

The LPLO corresponding to this generalised domain wall must append $\bar{Z}$ operators to $\bar{X}$ operators. Since $\bar{Z}_1$ is perpendicular to $\bar{X}_1$ and $\bar{Z}_2$ is perpendicular to $\bar{X}_2$, it must act as $\overline{CZ}: X_1 \to X_1 Z_2,\, X_2 \to Z_1X_2$. Such an operator can indeed be realised, by a plane of SWAP operators acting on each two-qubit site on the plane. The most straightforward plane to consider is an $xz$ plane, since it then intersects with the support of $\bar{X}_1$ on the $x$ axis, on which $\bar{Z}_2$ is supported, and with the support of $\bar{X}_2$ on the $z$ axis, on which $\bar{Z}_1$ is supported. Such a plane of SWAP operators fixes the logical $Z$ operators, $XX$ and $ZZ$, and maps $IX \to XI=IX \cdot XX$ and $YX \to XY = YX \cdot ZZ$.

Such a plane of SWAP operators oriented at any angle, however, will also implement the $\overline{CZ}$ operator. Of particular importance is the special case where the plane is oriented parallel to one of the $\bar{X}$ operators, say $\bar{X}_1$. Then, it appends a plane of $ZZ$ operators spanned by $\hat{x}$ and $\hat{y}+\hat{z}$. This may be viewed as a product of $L$ $\bar{Z}_2$ operators. Since $L$ is odd, $\bar{Z}^L=\bar{Z}$, and so we have that the operator appends $\bar{Z}_2$ to $\bar{X}_1$, as required. It intersects with $\bar{X}_2$ in the line spanned by $\hat{y}+\hat{z}$ and so appends a line of $ZZ$ operators oriented in the $\hat{y}+\hat{z}$ direction. Since this line traverses the entire lattice in the $z$ direction, this is another realisation of a $\bar{Z}_1$ operator, and so the operator also acts on $\bar{X}_2$ as required.

In summary, when $L$ is odd, we have shown that the code admits a $\overline{CZ}$ logical operator that may be realised by a plane of SWAP operators acting on each two qubit site in a plane of the lattice. 

In addition to this $\overline{CZ}$ logical operator acting within a single code, we may also consider if there are any operators that act between qubits in a stack of Levin-Wen codes. We consider one and two dimensional generalised domain walls in turn.

The first possible type is a one dimensional wall that would act trivially on all electric charges, but could append charges from a different code to fluxes. Such a wall would be analogous to the $s_{ij}$ walls in stacks of surface codes. Such a wall cannot be allowed here, however, since we have that $T_{e_im_j,e_im_j}= T_{e_i,e_i}S_{e_i,m_j}T_{m_j,m_j}=-1 \cdot 1 \cdot T_{m_j,m_j}=-T_{m_j,m_j} \neq T_{m_j,m_j}$, where we use that $S_{e_i,m_j}=1$ for $i\neq j$ and the identity in Eq.~\ref{eq:exchid}. This rules out $\overline{CZ}$ operators being allowed between qubits in different codes in the stack, and verifies that there are no allowed one dimensional generalised domain walls acting between different Levin-Wen codes.

We may now consider a two dimensional wall that appends charges from different codes to charges and fluxes from different codes to fluxes. Such a wall can also not be allowed in this code, however. To see this, assume $e_i \to e_ie_j$ for $i \neq j$. Then, by identity \ref{eq:exchid}, $T_{e_ie_j,e_ie_j}=T_{e_i,e_i}S_{e_i,e_j}T_{e_j,e_j}=T_{e_i,e_i} \cdot 1 \cdot -1=-T_{e_i,e_i}\neq T_{e_i,e_i}$. Thus, a wall that appends single charges to other charges cannot preserve exchange statistics of the code. This also rules out any wall that appends a single flux $m_j$ to $m_i$, since if $e_j$ is fixed by such a wall then we would have that braiding statistics of $m_i$ and $e_j$ are not preserved.

The only allowed type of wall acting between different codes is one that maps a single charge of one type to another single charge: $e_i \leftrightarrow e_j$ and, correspondingly, $m_i \leftrightarrow m_j$. Such a domain wall corresponds to $\overline{\text{SWAP}}_{(i,1),(j,1)}\overline{\text{SWAP}}_{(i,2),(j,2)}$, i.e.~SWAP operators acting between corresponding logical qubits of the two codes. This operator will clearly be realised as a three dimensional operator acting as a $\overline{\text{SWAP}}$ between each qubit in one code and its corresponding qubit in the other code.

So, in summary, we have shown that the only Clifford operators admitted by a Levin-Wen codes with an odd value of $L$ are those generated by a two dimensional $\overline{CZ}$ between pairs of qubits from the same code, and a three dimensional $\overline{\text{SWAP}}$ of corresponding qubits in different Levin-Wen codes.

We now consider possible non-Clifford LPLOs. Since the only Clifford excitation of codimension at least two is the $p$ excitation, we need only consider walls that involve such an excitation. There are two possible candidates for this type of wall. One possibility is that it could act as $m \to p$. This is not allowed, however, as $p$ must braid trivially with $e$ (since it acts as trivially on $e$ as a wall), while $m$ braids non-trivially with $e$.  The other possible wall we could consider is one that acts as $m \to mp$. This is also not allowed, however. To see this, note that $T_{p,p} \propto \overline{CZ}^2 \propto \bar{I}$ where the constant of proportionality is a phase, while $S_{m,p}=K(\overline{CZ},\bar{X}_1)=\bar{Z}_2$. Therefore, by the identity in Eq.~\ref{eq:exchid}, we have the following
\begin{align}
T_{mp,mp} &= T_{m,m}S_{m,p}T_{p,p} \nonumber \\
&\propto T_{m,m} \bar{Z}_2 \nonumber \\
&\neq T_{m,m}
\end{align}
So, the candidate wall does not preserve exchange statistics. Thus, there are no non-Clifford domain walls admitted by the code, and so also no non-Clifford LPLOs.

We now briefly consider the even-$L$ case. Recall that in this case there are three logical qubits. A set of three anticommuting pairs of logical Pauli operators may be realised as follows \cite{Haah}. Again, each logical operator is only specified up to a phase.
\begin{itemize}
\item $\bar{X}_1$ is realised by alternating lines of $YX$ and $ZI$ operators oriented parallel to $\hat{z}$, acting across a plane spanned by $\hat{y}$ and $\hat{z}$. $\bar{Z}_1$ is realised by $XX$ operators acting along a line oriented parallel to $\hat{x}$.
\item $\bar{X}_2$ is realised by alternating lines of $ZY$ and $IX$ operators oriented parallel to $\hat{x}$, acting across a plane spanned by $\hat{z}$ and $\hat{x}$. $\bar{Z}_2$ is realised by $YY$ operators acting along a line oriented parallel to $\hat{y}$.
\item $\bar{X}_3$ is realised by alternating lines of $XZ$ and $IY$ operators oriented parallel to $\hat{y}$, acting across a plane spanned by $\hat{x}$ and $\hat{y}$. $\bar{Z}_1$ is realised by $ZZ$ operators acting along a line oriented parallel to $\hat{z}$.
\end{itemize}

To find the LPLOs for the Levin-Wen fermion model with even $L$, we simply need to consider the two types of generalised domain walls we found in the case of odd-$L$ and reinterpret them as logical operators in light of these new logical Pauli operators. 

It is immediate to see that the $e_i \leftrightarrow e_j$, $m_i \leftrightarrow m_j$ wall that acts between Levin-Wen codes in a stack acts as SWAP operators between corresponding qubits in codes $i$ and $j$; that is as $\overline{\text{SWAP}}_{(i,1),(j,1)} \overline{\text{SWAP}}_{(i,2),(j,2)}\overline{\text{SWAP}}_{(i,3),(j,3)}$. The operator can clearly be realised again as a product of SWAP operators between corresponding qubits in codes $i$ and $j$ across the whole lattice.

The $p: m \to em$ wall now corresponds to three different LPLOs, depending on its orientation. Specifically, $p$ traversing a $yz$ plane appends a $YY$ string oriented in the $y$ direction to $\bar{X}_3$, and thus maps $\bar{X}_3 \to \bar{Z}_2\bar{X}_3$. Similarly, it appends a $ZZ$ string oriented in the $z$ direction to $\bar{X}_2$, and thus maps $\bar{X}_2 \to \bar{X}_2\bar{Z}_3$. Depending on whether it crosses in the $z$ or $y$ direction it also appends $L$ $\bar{Z}_2$ or $\bar{Z}_3$ operators to $\bar{X}_1$. Since $L$ is even, however, $\bar{Z}_2^L=\bar{Z}_3^L=I$, and so it acts trivially on $\bar{X}_1$. Thus, the wall corresponds to the $\overline{CZ}_{23}$ operator. Similarly, $p$ traversing the lattice in a $zx$ plane corresponds to $\overline{CZ}_{31}$ and in an $xy$ plane corresponds to $\overline{CZ}_{12}$ operator. The realisation of these operators is more complicated than for the case of even $L$. Specifically, the $\overline{CZ}_{12}$ is realised by a plane of sitewise $U_{12} = H^{\otimes 2} R_2^{\otimes 2} CZ\, H^{\otimes 2}$SWAP operators. $\overline{CZ}_{23}$ and $\overline{CZ}_{31}$ are realised by appropriate planes of $U_{23}= R_2^{\otimes 2} H^{\otimes 2} CZ H^{\otimes 2}$SWAP and $U_{31}=CZ\,R_2^{\otimes 2}$SWAP operators respectively.

\section{Finding all Locality-Preserving Logical Operators}
\label{V}

In this section, we present our detailed framework for achieving the primary aim of the paper: finding all of the LPLOs admitted by a topological stabiliser code that is locally equivalent to a finite number of copies of a $d$ dimensional surface code. This framework broadly consists of two components:
\begin{enumerate}
\item Identifying the full group of generalised domain walls admitted by the code;
\item Inferring the corresponding group of LPLOs from these generalised domain walls.
\end{enumerate}
Since the generalised domain walls of a code do not depend on its boundary conditions, we fully characterize the first step by considering stacked codes, to provide a classification of the generalised domain walls admitted by any code, in Sec.~\ref{VB}. The second step, however, depends on the choice of boundary conditions for which we do not have a full classification.  As such, for this step we provide only an outline of the approach, illustrated by examples, in Sec.~\ref{VC}.

We first pause to briefly consider the precise class of topological stabiliser codes to which we should expect our analysis to apply. While our focus has been primarily on codes that are locally equivalent to a finite number of toric codes, we observed through the example of the Levin-Wen fermion model in Sec.~\ref{Sec:LW} that there exist codes that lie outside this class but which may also be understood in our framework. Most generally, we should expect our analysis to be valid under the assumption that the eigenstate excitations of the code are of some well-defined, integer dimension that is fixed for all lattice sizes. This is necessary for our dimensional analysis of generalised domain walls to apply. Such an assumption is violated by codes with fractal structures to their excitations, such as those described in Refs.~\cite{HaahCode,BrellCode}. Provided that the excitations maintain their structure as the lattice is grown, or translated, however, the assumption will be satisfied. For this reason, we may conclude that our analysis will apply to stabiliser codes that are both translationally and scale invariant, referred to as \emph{STS codes} by Yoshida~\cite{YoshidaD}.

\subsection{General Constraints on Locality-Preserving Logical Operators}
\label{VA}
Before we develop our procedure, we first consider general properties of LPLOs. To build towards this, recall first the Clifford hierarchy on unitaries, defined recursively as follows,
\begin{align}
\mathcal{C}_1 &= \mathcal{P}_n \text{ (the Pauli group on $n$ qubits)}\\
\mathcal{C}_k &= \{U|UPU^\dag \in \mathcal{C}_{k-1} \forall P \in \mathcal{P}_n\} \text{ for $k>1$}.
\end{align}

That there exists a relationship between LPLOs and the Clifford hierarchy was first observed by Bravyi and K{\"o}nig \cite{BK}. Specifically, Bravyi and K{\"o}nig showed that all LPLOs admissable in a $d$ dimensional topological stabiliser code are in $\mathcal{C}_d$. This relationship is perhaps initially surprising. Within our framework, however, we may develop insight into how the Clifford hierarchy relates to LPLOs in codes, by considering the dimensionality of generalised domain walls and excitations, and using the correspondence between generalised domain walls and LPLOs we have established. This also allows us to derive a stronger version of Bravyi and K{\"o}nig's bound (as well as a stronger bound due to Yoshida and Pastawski \cite{PY}) which holds for a large class of topological stabiliser codes.

To build this insight, we begin by extending the Clifford hierarchy to apply to generalised domain walls and excitations, rather than only operators. Specifically, define a generalised domain wall to be in $\mathcal{C}_k$ if and only if its corresponding LPLO is in $\mathcal{C}_k$. We refer to such a generalised domain wall as a \textit{$\mathcal{C}_k$ domain wall}. Similarly, define a (non-eigenstate) excitation to be in the $k$th level of the hierarchy if and only if its corresponding generalised domain wall is in $\mathcal{C}_k$. Note that an immediate consequence of these definitions is that $\mathcal{C}_k$ domain walls are those that map eigenstate excitations to $\mathcal{C}_{k-1}$ excitations.

We now present our result. Recall from the introduction to Sec.~\ref{V} that this result applies to stabiliser codes that are translationally and scale invariant (STS codes). The proof of this result illuminates how it follows from the dimensional constraints on generalised domain walls discussed in the previous section.
\begin{theorem}
\label{Th1}
An LPLO in $D_p \equiv \mathcal{C}_{p} \setminus \mathcal{C}_{p-1}$, for $p \geq 2$, in an STS code of $d$ spatial dimensions, with minimum eigenstate excitation dimension, $a$, must have support of dimension, $m$, satisfying the following relation,
\begin{equation}
m \geq p(a+1).
\end{equation} 
\end{theorem}
Before proving this theorem, note that, since a $d$ dimensional code cannot admit an LPLO with support of dimension greater than $d$, the following corollary immediately follows.
\begin{cor}
\label{cor1}
A $d$-dimensional STS code can support an LPLO from $\mathcal{C}_{p} \setminus \mathcal{C}_{p-1}$ only if
\begin{equation}
p \leq \frac{d}{a+1}.
\end{equation}
\end{cor}
This immediately implies the Bravyi-K{\"o}nig bound for our codes, since $a \geq 0$. Also, note the distance of a topological code scales as the minimum dimension of a logical operator.  This minimum dimension is greater (by one) than the minimum dimension of an eigenstate excitation, i.e.,~it is minimum dimension $a+1$. Thus, corollary 1 also implies that the distance, $\delta$, of a $d$-dimensional STS code which admits an LPLO in $\mathcal{C}_p\setminus \mathcal{C}_{p-1}$ must satisfy the following relation,
\begin{equation}
\delta \leq O\left(\frac{d}{p}\right).
\end{equation}
This also implies the distance tradeoff theorem of Pastawski and Yoshida \cite{PY} for these codes.
We now prove theorem \ref{Th1}.
\begin{proof}[Proof of Theorem \ref{Th1}]
We show by induction on $p$ that a $D_p$ domain wall must be of dimension $k_p$ satisfying the following relation,
\begin{equation} \label{Toprove}
 k_p \geq p(a+1)-1.
\end{equation}
The result then follows since a $D_p$ LPLO must correspond to a $D_p$ domain wall of one less dimension than itself. 

Recall first, from Sec.~\ref{IVA2}, that the dimension, $k$ of a generalised domain wall that transforms an $l$ dimensional region of a $j$ dimensional excitation in a $d$ dimensional code is given by the following,
\begin{equation}
k=d-1-j+l \label{DWeq}.
\end{equation}
For $p=2$, we require that the LPLO corresponds to a Clifford domain wall. Thus, the wall must permute eigenstate excitations. Note that if the minimum eigenstate excitation dimension is $a$, the maximum dimension of an eigenstate excitation must be $d-2-a$. Thus, a Clifford domain wall must transform a region of dimension at least $a$ within an excitation of dimension at most $d-2-a$. Thus, Eq.~\ref{DWeq} implies that the domain wall must be of dimension $k_1$ satisfying the following relation,
\begin{equation}
k_2\geq d-1-(d-2-a)+a= 2(a+1)-1.
\end{equation}
Thus, the result holds for $p=1$.

Assume now that a $D_{p-1}$ domain wall is of dimension $k_{p-1}$ satisfying the following condition,
\begin{equation}
k_{p-1} \geq (p-1)(a+1)-1.
\end{equation}
This implies that the dimensions of $D_{p-1}$ excitations must satisfy the same constraint. A $D_p$ domain wall must map an eigenstate excitation to some $D_{p-1}$ excitation. Thus, it must transform a region of dimension at least $k_{p-1}$ within an excitation of dimension at most $d-2-a$. Thus, using Eq.~\ref{DWeq} and simplifying, the dimension, $k_p$, of a $D_p$ domain wall must satisfy Eq.~\ref{Toprove} as required.
\end{proof}
With this constraint on LPLOs derived, we now proceed to consider how the operators themselves may be identified for particular codes.

\subsection{Full Classification of Generalised Domain Walls}
\label{VB}

\subsubsection{Outline of Approach}
The procedure we follow to find the generalised domain walls admitted by a code locally equivalent to a finite number of toric codes draws together a number of the ideas we have considered so far. Firstly, we note that we need only consider stacked codes, as the generalised domain walls are independent of the choice of boundary conditions. Secondly, we use the correspondences we have identified between generalised domain walls, LPLOs and excitations. Specifically, for stacked codes we have identified that there is a bijective correspondence between $k$ dimensional generalised domain walls in the $n$th level of the Clifford hierarchy and (non-eigenstate) excitations of the same dimension and same level of the hierarchy. Again for stacked codes, we also have a further bijective correspondence of $k+1$ dimensional locality preserving logical gates in the $n$th level of the Clifford hierarchy with $k$ dimensional generalised domain walls in the same level of the hierarchy. We use both of these correspondences in our procedure.

The procedure is as follows.
\begin{enumerate}
\item Set $n=1$.
\item Identify $\mathcal{C}_{n+1}$ domain walls. This is done by considering, in turn, generalised domain walls of each dimension from 1 up to $d-1$. For each dimension, we find all mappings from eigenstate excitations to $\mathcal{C}_{n}$ excitations which preserve braiding and exchange relations, and satisfy the dimensional constraint on generalised domain walls found in Sec.~\ref{IVA2}. If the set of $\mathcal{C}_{n+1}$ domain walls is the same as the set of $\mathcal{C}_n$ domain walls then all generalised domain walls have already been found, so terminate the process. 
\item Relate this set of $\mathcal{C}_{n+1}$ domain walls to a corresponding set of $\mathcal{C}_{n+1}$ excitations by using the dimension-preserving, bijective correspondence between them.
\item Relate the set of $\mathcal{C}_{n+1}$ domain walls to the group of $\mathcal{C}_{n+1}$ LPLOs by using the bijective correspondence between them.
\item Use the correspondence, induced by the correspondences from the previous two steps, between $\mathcal{C}_{n+1}$ excitations and $\mathcal{C}_{n+1}$ LPLOs to determine the braiding and exchange relations of $\mathcal{C}_{n+1}$ excitations by using the group commutators and squares of the $\mathcal{C}_{n+1}$ logical operators.
\item Set $n\to n+1$ and return to step 2.
\end{enumerate}
We note that corollary \ref{cor1} ensures that this procedure will indeed terminate at some finite value of $n$.  We also note that, as a result of step 4, the procedure also gives the corresponding logic gates as LPLOs admitted by a stacked code.

We now apply this method to a general stacked code to fully classify the allowed generalised domain walls in any code locally equivalent to a finite number of toric codes, using induction to see the consequences of applying it on stacks of arbitarily many codes of arbitrarily high dimension.

\subsubsection{Classification for Identical Toric Codes}
We illustrate the classification of generalised domain walls first for the special case where the toric codes are identical (i.e.~have eigenstate excitations of the same dimension as one another). We explain how this may easily be generalised to the case of non-identical codes in Sec.~\ref{VB3}. The classification is as follows.

\begin{theorem}\label{Th2}
The generalised domain walls admitted by a code locally equivalent to $n$ identical copies of a $d$ dimensional toric code, with magnetic fluxes of dimension $M \geq \frac{d}{2}-1$ are products of:
\begin{itemize}
\item $h_i: e_i \leftrightarrow m_i$, of dimension $d-1$, iff $M=\frac{d}{2}-1$
\item $c_{ij}: m_i \to m_im_j, \, e_j \to e_i e_j$, of dimension $d-1$, iff $n \geq 2$
\item $s^{(k)}_{i_1, \ldots i_k}: m_{i_a} \to m_{i_a}s^{(k-1)}_{i_{a+1}, \ldots i_{k-1},i_1,\ldots i_{a-1}}$, of dimension $k(d-M-1)-1$, for all $k$ such that $k \leq \min\left(n,\frac{d}{d-M-1}\right)$.
\end{itemize}
\end{theorem}

As discussed, this classification also immediately gives us a classification of all logical operators that may be realised as LPLOs for a stack of identical toric codes. Note that in this classfication we refer to an LPLO with support on a $k$-dimensional subspace as being of dimension $k$. We also denote by $C^{k-1}Z$ a $Z$ operator controlled by $k-1$ qubits. The classification can then be summarised as follows.
\begin{cor}
The LPLOs admitted by a $d$ dimensional stacked code consisting of $n$ identical toric codes with magnetic fluxes of dimension $M\geq \frac{d}{2}-1$ are products of Pauli operators and the following.
\begin{itemize}
\item $H$, of dimension $d$, iff $M=\frac{d}{2}-1$
\item $CNOT$, of dimension $d$, iff $n \geq 2$
\item $C^{k-1}Z$, of dimension $k(d-M-1)$, for all $k$ such that $k \leq \min\left(n,\frac{d}{d-M-1}\right)$
\end{itemize}
\end{cor}
We note that this result is consistent with the known LPLOs admitted by a colour code with $M=d-2$ on a cubic lattice, which is equivalent to a stacked code where $n=d$ \cite{Kubica}. Note also that the assumption that $M \geq \frac{d}{2}$ is not of great significance, since in a case where $M < \frac{d}{2}$ we can simply substitute $E$ for $M$ in the dimensions of LPLOs and walls, and conjugate each LPLO by Hadamards on every code. We defer including this explicitly to our more general classification in Sec.~\ref{VB3}.

To work towards proving the theorem, we first prove three lemmas. The first concerns the special case where $M=E=\frac{d}{2}-1$. We show that the generalised domain walls allowed here correspond to those we found for two dimensional codes.

\begin{lem}\label{l1}
The generalised domain walls admitted by a code locally equivalent to $n \geq 2$ identical copies of a $d$ dimensional toric code, with magnetic fluxes and electric charges of the same dimension, $M=E=\frac{d}{2}-1$ are products of the following $d-1$ dimensional domain walls,
\begin{align}
h_i: & \, e_i \leftrightarrow m_i\\
s_{ij}^{(2)}: & \, m_i \to m_ie_j,\, m_j \to e_im_j
\end{align}
\end{lem}
\begin{proof}
Since all eigenstate excitations are equal in dimension then, by Eq.~\ref{DWeq} any allowed permutation of these excitations corresponds to a domain wall of dimension $d-1$. This implies that dimensional constraints do not play any role in restricting the allowed domain walls in this case, and so the group of allowed Clifford domain walls in any code of this type will be the same as in the two dimensional stacked code. We have already seen that this is the group generated by $h_i$ and $s_{ij}$, which correspond to Hadamard and $\overline{CZ}$ LPLOs respectively in the stacked code. Since all Clifford domain walls, and hence all Clifford excitations, are $d-1$ dimensional, they are all of higher dimension than the eigenstate excitations. Thus, no $\mathcal{C}_3 \setminus \mathcal{C}_2$ domain walls are possible. Hence, the full group of domain walls for a code where $M=E$ is generated by $h_i$ and $s_{ij}^{(2)}$.
\end{proof}

We now begin to consider the more complicated case, where $M>E$, starting with Clifford and $\mathcal{C}_3$ domain walls. We show that these walls correspond to those we found for three dimensional codes.
\begin{lem}\label{l2}
The $\mathcal{C}_3$ domain walls (including Clifford domain walls) admitted by a code locally equivalent to $n \geq 3$ identical copies of a $d$ dimensional toric code, with magnetic fluxes of dimension $M>\frac{d}{2}-1$ are products of the following generalised domain walls:
\begin{align}
c_{ij}: & \, m_i \to m_im_j, \, e_j \to e_i e_j\\
s_{ij}^{(2)}: & \, m_i \to m_ie_j,\, m_j \to e_im_j\\
s_{ijk}^{(3)}: & \, m_i \to m_i s_{jk}^{(2)}, \, m_j \to m_j s_{ki}^{(2)}, \, m_k \to m_k s_{ij}^{(2)}
\end{align}
with $c_{ij}$ of dimension $d-1$, $s_{ij}^{(2)}$ of dimension $2(d-M-1)-1$ and $s_{ijk}^{(3)}$ of dimension $3(d-M-1)-1$.
\end{lem}
\begin{proof}
We consider first Clifford domain walls, and note that these correspond to those admitted by the three dimensional stacked code, discussed in Sec.~\ref{IVB2}. To see this observe that Clifford domain walls here may either permute eigenstate excitations of the same dimension, or append an $E$ dimensional excitation to an $M$ dimensional one. By Eq.~\ref{DWeq}, the dimensions of walls of these types will be $d-1$ and $2(d-M)-3$ respectively. For the three dimensional stacked code (where $d=3$, $M=1$) this corresponds to walls of dimension $2$ and $1$ respectively. Thus, any Clifford domain wall which is permitted by braiding and exchange statistics will also be allowed in the three dimensional stacked code since it is not excluded by dimensional constraints. So, since any wall that does not satisfy these braiding and exchange constraints will not be permitted for any stacked code, the Clifford domain walls allowed for any  code where $M>E$ will indeed be the same as for the three dimensional stacked code. That is, the Clifford domain walls are generated by the following generalised domain walls:
\begin{align}
c_{ij}: & \, m_i \to m_im_j, \, e_j \to e_i e_j\\
s_{ij}^{(2)}: & \, m_i \to m_ie_j,\, m_j \to e_im_j
\end{align}
with $c_{ij}$ of dimension $d-1$ and $s_{ij}^{(2)}$ of dimension $2(d-M-1)-1$.

We may now consider $\mathcal{C}_3$ domain walls. Recall that these are the walls which correspond to LPLOs in the third level of the Clifford hierarchy. Since $c_{ij}$ is of dimension higher than any eigenstate excitation, it cannot be involved in any $\mathcal{C}_3$ domain wall. Thus, the only types of $\mathcal{C}_3 \setminus \mathcal{C}_2$ domain walls possible must map $m$ type excitations to excitations involving $s^{(2)}$. As demonstrated in the case of the three dimensional stacked code, the only domain wall of this form allowed is
\begin{equation}
s_{ijk}^{(3)}: m_i \to m_i s_{jk}^{(2)}, \, m_j \to m_j s_{ki}^{(2)}, \, m_k \to m_k s_{ij}^{(2)}
\end{equation}
By Eq.~\ref{DWeq}, this will be $3(d-M-1)-1$ dimensional, and so is allowed for any code where $\frac{d}{d-M-1}\geq 3$. So, to summarise, the only types of $\mathcal{C}_3$ domain walls possible are generated by 
\begin{align}
c_{ij} &\text{ of dimension } d-1\\
s_{ij}^{(2)} &\text{ of dimension } 2(d-M-1)-1\\
s_{ijk}^{(3)} &\text{ of dimension } 3(d-M-1)-1.
\end{align}
\end{proof}
This lemma is sufficient to illustrate the pattern of generalised domain walls at all levels of the Clifford hierarchy. This is summarised by our third and final lemma.
\begin{lem}\label{l3}
For $k \geq 3$, the $\mathcal{C}_k \setminus \mathcal{C}_{k-1}$ domain walls possible in a code locally equivalent to $n \geq k$ identical copies of a $d$ dimensional toric code are products of
\begin{align}
s_{i_1, \ldots , i_k}^{(k)}: m_{i_a} \to m_{i_a} s^{(k-1)}_{i_{a+1}, \ldots i_k, i_1, \ldots i_{a-1}}
\end{align}
with $\mathcal{C}_{k-1}$ domain walls. Moreover, $s_{i_1, \ldots , i_k}^{(k)}$ walls are of dimension $k(d-M-1)-1$.
\end{lem}
\begin{proof}
We begin by showing by induction that wall $s_{i_1, \ldots , i_k}^{(k)}$ corresponds to $C^{k-1}Z_{i_1, \ldots , i_k}$. This is simple, since we know that $s^{(2)}_{ij}$ corresponds to $CZ_{ij}$, and that, assuming $s^{(k-1)}_{i_{a+1}, \ldots i_k, i_1, \ldots i_{a-1}}$ corresponds to $C^{k-2}Z_{i_{a+1}, \ldots i_k, i_1, \ldots i_{a-1}}$ then $s_{i_1, \ldots , i_k}^{(k)}$ corresponds to the gate which maps $\bar{X}_{i_a} \to {X}_{i_a}C^{k-2}Z_{i_{a+1}, \ldots i_k, i_1, \ldots i_{a-1}}$ which is $C^{k-1}Z_{i_1, \ldots , i_k}$.

We now prove the proposition by induction on the level of the Clifford hierarchy, $k$. We have already shown, in lemma \ref{l2}, that it is true for $k=3$. So, begin by assuming that $\mathcal{C}_{k-1}$ domain walls are generated by $s^{(k-1)}$ walls and $\mathcal{C}_{k-2}$ walls. We first show that $s^{(k)}$ is an allowed wall, and then that all other $\mathcal{C}_k$ walls are generated by $s^{(k)}$ together with $\mathcal{C}_{k-1}$ walls.

To do the first step, note that $\left({X}_{i_a}C^{k-1}Z_{i_{a+1}, \ldots i_k, i_1, \ldots i_{a-1}}\right)^2=I$ and so $s^{(k)}$ preserves exchange relations. Also, since $C^{k-1}Z$ is a unitary operator then it preserves commutation relations, and so its corresponding domain wall $s^{(k)}$ must preserve brading relations. Thus, $s^{(k)}$ is an allowed domain wall, which maps an eigenstate excitation to a $\mathcal{C}_{k-1}$ excitation, and hence is a $\mathcal{C}_k$ domain wall which is allowed.

Now, to do the second step, observe that if we have a $\mathcal{C}_k \setminus \mathcal{C}_{k-1}$ domain wall it must map electric charges to other electric charges, since they are lower dimension than any other excitations. Thus, it must map some magnetic flux to another excitation which includes an $s^{(k-1)}$ excitation, since this is the only type of $\mathcal{C}_{k-1}$ excitation. Moreover, since $K(X_{i_1}, C^{k-1}Z_{i_1, \ldots i_{k-1}})=C^{k-2}Z_{i_2 \ldots i_{k-1}}$ which cannot be made by a product of $C^rZ$ gates for $r<k-2$, if any magnetic flux is mapped to another excitation involving $s^{(k-1)}$ then they all must be, in order to preserve braiding relations. Also, since all $C^rZ$ type gates commute with each other, then in order to preserve braiding relations each $m$ must map to an excitation involving $m$. Thus, we conclude that any $\mathcal{C}_k$ domain wall must map a magnetic flux to an excitation involving both $m$ and $s^{(k-1)}$ excitations. 

Moreover, the $m$ and $s^{(k-1)}$ must come from non-overlapping codes, or the excitation will have non-trivial exchange relations. However, we have already seen that the wall that maps $m$ to an excitation consisting of only this $m$ and $s^{(k-1)}$ is allowed, and indeed corresponds to $s^{(k)}$. So, if we have another $\mathcal{C}_k$ domain wall we can always apply $w_{ij}=c_{ij}c_{ji}c_{ij}$ type walls to swap the codes of the output of the wall such that the code from which the magnetic flux originally came remains the same, and then apply $s^{(k)}$ to give a $\mathcal{C}_{k-1}$ domain wall. However, this domain wall must already be included in the walls already found, and so must be expressible as a product of $s^{(r)}$ and $c_{ij}$ type excitations, with $r<k$. Thus, this new $\mathcal{C}_k$ must be expressible as a product of this same combination of excitations along with $s^{(k)}$. Thus, all the $\mathcal{C}_k$ domain walls can be produced from products of $\mathcal{C}_{k-1}$ domain walls with $s^{(k)}$ type walls.

We can verify that $s^{(k)}$ has dimension $k(d-M-1)-1$ as claimed by induction on $k$. First, note that $s^{(2)}$ has dimension $2(d-M-1)-1$ as required. Now, assume that $s^{(k-1)}$ has dimension $(k-1)(d-M-1)-1$. Then, $s^{(k)}$ must transform a $(k-1)(d-M-1)-1$ dimensional region of an $M$ dimensional excitation. By Eq.~\ref{DWeq}, this requires a domain wall of dimension $d-1-(M-((k-1)(d-M-1)-1))=k(d-M-1)-1$.
\end{proof}
We may now finally prove theorem \ref{Th2}, to complete our classification.
\begin{proof} [Proof of Theorem \ref{Th2}]
By lemmas \ref{l1} and \ref{l2}, we know that $h_i$ is admitted iff $M=\frac{d}{2}-1$, and in that case it has dimension $d-1$. Since $c_{ij}$ acts over two toric codes, it clearly has as a necessary condition that $n\geq 2$. Given this condition, it is admitted in codes with $M=\frac{d}{2}-1$, by lemma \ref{l1}, since it may be realised by $c_{ij}=h_j s^{(2)}_{ij} h_j$. By lemma \ref{l2}, $c_{ij}$ is also admitted by codes with $M>\frac{d}{2}-1$, provided that $n \geq 2$. Thus, $c_{ij}$ is admitted iff $n \geq 2$, and, from lemmas \ref{l1} and \ref{l2}, always has dimension $d-1$.

Since $s^{(k)}$ type domain walls act across $k$ different toric codes, it is a necessary condition for its admission that $n\geq k$. Assuming this condition then, by lemmas \ref{l1} and \ref{l2}, $s_{ij}^{(2)}$ is always admitted, which is consistent with the claim of the theorem, since $d-M-1 \leq d-(\frac{d}{2}-1)-1=\frac{d}{2}$, and so $\frac{d}{d-M-1} \geq \frac{d}{d/2}=2$, and so the condition, $k \leq \frac{d}{d-M-1}$ is automatically satisfied for $k=2$. Moreover, by lemma \ref{l1}, when $M=\frac{d}{2}-1$, the dimension of the $s^{(2)}$ domain wall is $d-1$, which is equal to $2(d-(\frac{d}{2}-1))-1$, and so the claimed dimension of the wall in this case is correct. In the case where $M>\frac{d}{2}-1$, lemma \ref{l2} gives that we have the correct dimension for $s^{(2)}$ of $2(d-M-1)-1$.

For $k>2$, assume again the necessary condition that $n \geq k$. Then, we have from lemma \ref{l3} that a necessary and sufficient condition for $s^{(k)}$ to be admitted is that $d-1 \geq k(d-M-1)-1$, since the only additional requirement is that the largest dimension of a generalised domain wall admitted in the code is at least as great as the required dimension of the $s^{(k)}$. This implies the necessary and sufficient condition that $k \leq \frac{d}{d-M-1}$, as required. Moreover, lemma \ref{l3} gives that the dimension of $s^{(k)}$ is indeed $k(d-M-1)-1$.

Since lemmas \ref{l1},\ref{l2} and \ref{l3} collectively provide a full analysis of all possible domain walls, and we have verified that the listed walls are indeed admitted subject to the claimed conditions, then theorem \ref{Th2} does indeed account for all the possible generalised domain walls.
\end{proof}

\subsubsection{Generalised Classification}
\label{VB3}
The classification developed above may be generalised to the case of non-identical toric codes in a straightforward way. To see this, first note that the only differences between the toric codes will be the dimensions of excitations; there is no difference in terms of braiding or exchange relations. This means that we should expect the same types of generalised domain walls as in theorem \ref{Th2}, but with different dimensions, and hence different constraints on when those walls can be realised in a code.

Indeed, the correct generalisation of theorem \ref{Th2} is as summarised in theorem \ref{Th3} below. This result may be proven by a straightforward generalisation of the proof of theorem \ref{Th2}, which we omit for clarity and brevity.

\begin{theorem}\label{Th3}
Consider a code locally equivalent to $n$ toric codes, each of dimension $d$, such that code $i$ has magnetic and electric excitations of dimensions $M_i$ and $E_i$ respectively. Let $q_i = \max\left(M_i, E_i\right)$ and $x_i=1$ if $E_i > M_i$ or $0$ otherwise. Then, the generalised domain walls are products of:
\begin{itemize}
\item $h_i$, of dimension $d-1$, for all $i$ such that $M_i=\frac{d}{2}-1$
\item $h_i^{x_i}h_j^{x_j} c_{ij}h_i^{x_i}h_j^{x_j}$, of dimension $d-(q_i-q_j)-1$, for all $i,j$ such that $q_i \geq q_j$
\item $\left(\prod_{a=1}^k h_{i_a}^{x_{i_a}}\right) s^{(k)}_{i_1, \ldots i_k} \left(\prod_{a=1}^k h_{i_a}^{x_{i_a}}\right)$, of dimension $\sum_{a=1}^k \left(d-q_{i_a}-1\right)-1$, for all $k, i_1, \ldots i_k$ such that $\sum_{a=1}^k q_{i_a} \geq (k-1)d-k$
\end{itemize}
\end{theorem}
Note first that if we assume that $M_i \geq E_i$ for all $i$ then this theorem differs from theorem \ref{Th2} only by the dimensions of $c$ and $s^{(k)}$ being generalised. We have relaxed that assumption, however, since with non-identical codes there is a possibility of having some codes with $M_i\geq E_i$ and others with $E_i < M_i$. This leads to non-trivial differences from the case where $M_i \geq E_i$ for all codes, which is dealt with by conjugating walls on any code where $E_i < M_i$ by the wall which swaps $e_i$ and $m_i$.

As with theorem \ref{Th2}, this classification of generalised domain walls also naturally gives the LPLOs admitted by a general stacked code as follows.
\begin{cor}\label{cor3}
The LPLOs admitted by a stacked code with parameters as in theorem \ref{Th3} are products of Pauli operators and the following.
\begin{itemize}
\item $H_i$, of dimension $d$, iff $M_i=\frac{d}{2}-1$
\item $H_i^{x_i}H_j^{x_j}CNOT_{ij}H_i^{x_i}H_j^{x_j}$, of dimension $d-(p_i-p_j)$, for all $i,j$ such that $q_i \geq q_j$
\item $\left(\prod_{a=1}^k H_{i_a}^{x_{i_a}}\right)C^kZ_{i_1, \ldots i_k}\left(\prod_{a=1}^k H_{i_a}^{x_{i_a}}\right)$, of dimension $\sum_{a=1}^k(d-q_{a_i}-1)$, for all $k$ such that $\sum_{a=1}^k q_{i_a} \geq (k-1)d-k$.
\end{itemize}
\end{cor}

\subsubsection{Abelian Quantum Double Models}
The approach and results we have outlined in this section may be generalised beyond the topological stabiliser codes that are the focus of this paper. In particular, we briefly consider codes that are quantum doubles of abelian groups \cite{KitaevA,YoshidaC}. 

The quantum double of a group $G$ has generalised logical Pauli operators such that $\bar{X}$ and $\bar{Z}$ type operators each form groups isomorphic to $G$. It then has corresponding groups of electric charges and magnetic fluxes corresponding to these generalised Paulis operators. A stacked code with $n$ toric codes is the quantum double of $\mathbb{Z}_2^n$. We can also consider, however, the quantum double of an arbitrary abelian group $G=\prod_{i=1}^n \mathbb{Z}_{p_i^{r_i}}$ for $n,r_i \in \mathbb{N}$ and primes $p_i$. With appropriate boundary conditions, we can consider each cyclic group $\mathbb{Z}_{p_i}$ to give a corresponding generalised surface code. By generalised surface code we mean that it is equivalent to a surface code but with qudits (for $d=p_i$) instead of qubits. We can thus consider the code to be a stacked code made up of these $n$ disjoint generalised surface codes. With this set up, the results of theorem \ref{Th3} still holds, but with the definition of each wall adapted to ensure the altered braiding and exchange relations of excitations in the code are preserved, and with the additional constraint that walls across multiple codes are only allowed between codes with the same $p_i$. This naturally also gives rise to an analagous group of LPLOs to those in corollary \ref{cor3}.

\subsection{Logical Operators from Domain Walls}
\label{VC}
Equipped with a classification of the generalised domain walls possible in codes equivalent to a finite number of copies of a toric code, we may now discuss how to adjust corollary~\ref{cor3} to give the logic gates that are realisable as LPLOs for codes with boundary conditions that make them distinct from stacked codes. These boundary conditions manifest themselves as alterations to the one-to-one relationship between excitations and LPLOs that exists for stacked codes. We consider two types of alterations to produce different boundary conditions. One type is attaching multiple surface codes along a common boundary.  This can lead to excitations that are distinct in the stacked code giving rise to equivalent logical operators in the code with these boundaries. Examples of this are colour codes that encode a single qubit \cite{Kubica}.~The other type is to allow for holes to be added to the surface codes. These boundaries can allow for multiple topologically distinct paths of excitations, and so can lead to multiple logical operators arising from the same excitation. An example of this type of alteration is the toric code (in its standard formulation on a torus). We consider each of these cases in turn.

\subsubsection{Attaching Surface Codes}
We first consider the LPLOs that may arise from a code where surface codes in a stacked code are attached along a boundary. We lay out a procedure for finding these logical operators. At each stage we illustrate how this may be done with the example of the $d$ dimensional colour code on a hypertetrahedral lattice. This code is equivalent to $d$ surface codes attached along a common $d-1$ dimensional boundary~\cite{Kubica}. Our procedure is as follows.

\paragraph*{Eigenstate Excitation Equivalences:} Consider two eigenstate excitations to be equivalent if and only if they may be interchanged at a transparent boundary (i.e.~a boundary at which two surface codes are attached). Define a reduced set of eigenstate excitations by quotienting by this equivalence relation. We may relate each of these equivalence classes with a Pauli logical operator produced by a path followed by any representative of the class. For the colour code, we have boundaries that allow any electric charges to be interchanged, and similarly for any magnetic fluxes. Thus, we have $e_i \sim e_j$ and $m_i \sim m_j$ for $1\leq j,k \leq d$. We label the two equivalence classes of excitations by $e$ and $m$ respectively. They correspond to $\bar{Z}$ and $\bar{X}$ logical operators respectively. Naturally, we also have equivalence classes $1$ and $em$, that correspond to $\bar{I}$ and $\bar{Y} \propto \bar{X}\bar{Z}$.

\paragraph*{Clifford Domain Walls:}Consider the Clifford domain walls for the code. Consider two such walls to be equivalent if and only if they differ by composition with the action of a transparent boundary. We thereby define a reduced set of Clifford domain walls by quotienting by this equivalence relation. For the colour code, we have that boundaries allow the codes from which excitations come to be permuted. So, Clifford domain walls are equivalent if and only if their action on excitations differs only by the codes of these excitations, not their type (i.e.~$e$, $m$, etc.). For the colour code, this gives us up to two non-trivial equivalence classes of Clifford domain walls. Specifically, we have $s_{ij}^{(2)} \sim s_{kl}^{(2)}$ for $1 \leq i,j,k,l \leq d$, which gives us an equivalence class we label $s^{(2)}$. For codes where $h_i$ type walls are admitted, we also have $h_i \sim h_j$ for $1 \leq i,j \leq d$, which gives us a second equivalence class we label $h$. Note that the equivalence class of $c_{ij}$ type walls is trivial, since $c_{ij}$ only permutes elements of each equivalence class of eigenstate excitations, i.e.~it maps electric charges to electric charges and magnetic fluxes to magnetic fluxes. This means that $c_{ij} \sim 1$.

\paragraph*{Clifford LPLOs:} Consider a representative of each equivalence class of Clifford domain walls. The corresponding LPLO for this class is then identified by considering the action of the representative wall on eigenstate excitations from the surface codes the wall acts non-trivially on. Specifically, the operator is that which acts by conjugation to give the logical operator mappings corresponding to replacing the eigenstate excitations with the Pauli logical operators corresponding to their equivalence classes. For the colour code, we first consider the (possible) equivalence class $h$. This has representative $h_i: e_i \leftrightarrow m_i$ and so acts as $h: e \leftrightarrow m$. Thus, $h$ corresponds to the LPLO $\bar{H}: \bar{X} \leftrightarrow \bar{Z}$. That is, it is a logical Hadamard operator on the single logical qubit. Consider now the other equivalence class of Clifford domain walls, $s^{(2)}$. This has representative $s_{ij}^{(2)}: m_i \to m_i e_j$, $m_j \to e_i m_j$. Thus, it acts as $s: m \to em$ and so corresponds to an LPLO $\bar{R}_2: \bar{X} \to \bar{Y}$. That is, it is a logical phase operator on the single logical qubit. Note that, as required, the logical operators we find do not depend on the particular representatives of the equivalence classes that we choose. This is because the transparent boundaries correspond to symmetries of the underlying stacked code, and so the equivalences of eigenstate excitations carry through appropriately to equivalences of the Clifford domain walls.

\paragraph*{$\mathcal{C}_3$ Domain Walls:} We now consider each equivalence class of Clifford domain walls to be a Clifford excitation of the code. A reduced set of $\mathcal{C}_3$ domain walls is then produced by considering as equivalent $\mathcal{C}_3$ domain walls for the stacked code that act equivalently on this reduced group of Clifford excitations. This set of domain walls is related to LPLOs in an analogous way to that for Clifford domain walls. For the colour code, we now have two non-eigenstate, Clifford excitations; $h$ and $s^{(2)}$. All the possible new $\mathcal{C}_3$ domain walls differ only by the codes from which the excitations they act on come. Specifically, $s_{ijk}^{(3)} \sim s_{lmn}^{(3)}$ for $1 \leq i,j,k,l,m,n \leq d$. Thus, we have a single equivalence class of $\mathcal{C}_3 \setminus \mathcal{C}_2$ domain walls, which we label $s^{(3)}$. Since any $s^{(3)}$ type wall maps $m$ type excitations to $ms^{(2)}$ type excitations, we may conclude that the LPLO corresponding to $s^{(3)}$ is $\bar{R}_3: \bar{X} \to \bar{X}\bar{R}_2$.

\paragraph*{Iteration:} The process is repeated for increasing levels of the Clifford hierarchy until all of the generalised domain walls admitted by the code are exhausted. For the colour code, continuing this process will give us a new equivalence class of generalised domain walls at each level of the Clifford hierarchy, up to the $d$th level. These classes are of the form $s^{(k)}: m \to ms^{(k-1)}$, and correspond to $\bar{R}_k: \bar{X} \to \bar{X}\bar{R}_{k-1}$ logical operators.

Thus, we have an approach to finding the LPLOs admitted in a code by using the set of generalised domain walls we identified in theorem \ref{Th2}. Our analysis of the $d$-dimensional colour code on a hypertetrahedral lattice also yields the following result, consistent with the analysis of \cite{Kubica}, which may be considered a corollary to the classification of generalised domain walls in theorem \ref{Th2}.

\begin{cor}
The $d$-dimensional colour code on a tetrahedral lattice, encoding a single logical qubit, admits LPLOs generated by Pauli operators, the Hadamard operator if and only if $d=2$, and $\bar{R}_k$ operators for all $k \leq d$.
\end{cor}

\subsubsection{Adding Holes}
Consider the impact of producing new boundaries by adding holes to surface codes. With holes, instead of having equivalences between different excitations that correspond to the same logical operators, we have correspondences between different logical operators that correspond to different excitations. In particular, for every topologically distinct path ending in a boundary which a given excitation may take we have a different Pauli logical operator corresponding to this same excitation. Thus, generalised domain walls that act on the excitation will correspond to LPLOs that must act in the same way on all these Pauli logical operators.

To illustrate this, consider first a single $d$-dimensional toric code, realised on a d-torus, equivalent to a direct product of $d$ circles, or to a $d$ dimensional hypercube with opposite faces identified. Such a code with $E$-dimensional electric charges encodes ${d}\choose{E+1}$ logical qubits; since we may choose that many different combinations of $E+1$ circles to be included in the support of each logical $\bar{Z}_i$ operator. The other $d-(E+1)=M+1$ circles must be included in the support of the corresponding logical $X_i$ operator. Since each $\bar{X}$ (or $\bar{Z}$) operator is produced by the same magnetic flux (or electric charge), the $h: e \leftrightarrow m$ domain wall, when it is admitted (i.e.~when $E=M$), must correspond to a logical operator that acts non-trivially on all the $\bar{X}$ and $\bar{Z}$ operators. Specifically, this logical operator will be a product of Hadamard operators on all of the logical qubits, along with $\frac{1}{2}{{d}\choose{E+1}}=\frac{1}{2} {{d}\choose{d/2}}$  SWAP operations that reorder the logical qubits by swapping each qubit with $\bar{X}$ operator acting around a particular $d/2$ circles with the qubit which has $\bar{Z}$ operator acting around those same circles.

More generally, if we have a stack of such $d$ dimensional toric codes, we may consider the other domain walls we identified for the $d$ dimensional stack of surface codes. The $c_{ij}$ wall will correspond to the logical operator that is a product of $CNOT$ operators between each pair of corresponding qubits in toric codes $i$ and $j$. The $s_{ij}^{(2)}$ wall will correspond to a product of $d\cdot {{M+1}\choose {E+1}}$ logical $CZ$ operators between each qubit $\alpha$ in code $i$ and the ${M+1}\choose {E+1}$ qubits in code $j$ for which the corresponding circles to those included in the support of $\bar{Z}_{j,\beta}$ are also included in the support of $\bar{X}_{i,\alpha}$. Similarly, the $s_{ijk}^{(3)}$ wall will correspond to a product of logical $\overline{CCZ}$ operators, and so on for higher levels of the Clifford hierarchy.

The conclusion of this is that adding holes can add more logical qubits to a stacked code, but does not add more LPLOs. Instead, the LPLOs become (in general quite complicated) operators that act consistently with the topology of different logical operators. Note that this observation is similar to the concept of \textit{homology-preserving operators}, discussed for two dimensional codes by Beverland $\textit{et al}$ \cite{Beverland}.

\section{Conclusion}

In this paper, we have developed a framework for finding LPLOs in topological stabiliser codes. Using this approach, we have provided a full classification of the LPLOs admitted by codes that are locally equivalent to a finite number of disjoint surface codes in $d$ dimensions. In addition, we have illustrated how the results may be adapted to determine the LPLOs in codes that are locally equivalent to a number of surface codes in the bulk, but have other types of boundary conditions, specifically considering colour codes and generalised toric codes.  We have provided analysis of models that are not equivalent to any number of surface codes but which have a similar structure, specifically, the Levin-Wen fermion model and abelian quantum double models.  

This framework may assist with finding low-overhead architectures for implementing quantum algorithms.  Much of the overhead associated with standard architectures for fault-tolerant quantum computing can be attributed to the sophisticated techniques such as magic state distillation that are needed when locality-preserving implementations of logical operators are not available~\cite{BraHaah, Fowleretal}.  Our approach can be used to determine which logical operators are locality-preserving in a topological stabiliser code, and so offers the potential to match algorithms with codes that maximise the use of locality-preserving gates.  

While we have considered a large and widely-studied class of topological codes, we note that many important codes remain outside the scope of our analysis. Specifically, we have not given consideration to any models with non-abelian excitations. If our approach could be adapted to such models, it would be particularly interesting to compare the power of LPLOs in abelian and non-abelian models in higher dimensions. In particular, Beverland~\textit{et al.}~\cite{Beverland} have observed an apparent tradeoff in two dimensional codes between the computational power offered by the braiding of non-abelian anyons and the range of LPLOs admitted by codes with such anyons.  Determining whether such a tradeoff generalises to higher dimensional codes is an interesting open question.

Our analysis has also not considered the additional computational structures that are enabled through the introduction of topological defects, such as twists~\cite{BombinTwist, BombinTwist2, Brown}. The group of LPLOs admitted by a code is, in general, different from that admitted by the braiding of defects in such codes. For example, while the two dimensional surface code does not allow an $\bar{R}_2$ LPLO, such an operator can be realised by including twists~\cite{BombinTwist, Brown}. Thus, classifying the operators given by braiding twists is yet another worthwhile research direction in understanding the power of topological codes for fault tolerant quantum computation, beyond what we have considered here. 

Understanding the set of logical operators that can be implemented by braiding twists in a code is closely tied to this work, since twists exist as endpoints (or more generally, the boundaries) of domain walls~\cite{BombinTwist, Barkeshli}. Specifically, we may find a twist corresponding to the endpoint of a domain wall in a two dimensional code by allowing the domain wall to terminate in the bulk~\cite{Bridgeman}.  This necessitates altering the Hamiltonian at these endpoints so that all of the terms of the Hamiltonian continue to commute.  In general, as illustrated for the case of the two-dimensional surface code, the inclusion of domain walls with such topological defects at their boundaries may allow for a richer and potentially more powerful set of locally-implementable logic gates, and this warrants further study.

A natural next step would be to seek to understand analogous structures to twists in higher dimensional topological stabiliser codes. For example, we have seen that the three dimensional colour code admits a two dimensional domain wall, corresponding to the LPLO $\bar{R}_3$.  On an open surface, this domain wall possesses a loop-like boundary forming an extended topological defect, which may then admit interesting braiding relations. Indeed, one might wonder if the fact that this defect corresponds to a non-Clifford LPLO may allow for non-Clifford logical operators to be performed by braiding it. It would also be valuable to consider point-like defects that would exist at the endpoints of one dimensional, generalised domain walls in higher dimensional codes. While the braiding of such point-like defects in more than two dimensions must be trivial, they may admit interesting braiding relations with higher dimensional excitations and defects.  Finally, as yet another example worth further study, the three-dimensional model studied by Roberts \textit{et al.}~\cite{Roberts} is a domain wall in the four-dimensional toric code, and is thermally stable.  The topological defect associated with the two-dimensional boundary of this object is then also thermally stable, and as a result may offer additional robustness or other features to a computational model.  Understanding the nature of these defects in higher dimensional codes would be invaluable in understanding the computational power, and structure, of higher dimensional codes.

\begin{acknowledgments}
We acknowledge discussions with Jacob Bridgeman, Ben Brown, Andrew Doherty, Steve Flammia, Sam Roberts and Dominic Williamson.  This work was financially supported by the ARC via the Centre of Excellence in Engineered Quantum Systems (EQuS), project number CE110001013, and via project DP170103073. 
\end{acknowledgments}

\end{document}